\documentclass[%
twocolumn,
superscriptaddress,
 amsmath,amssymb,
 aps,
 pra,
floatfix
]{revtex4-2}

\usepackage{graphicx}
\usepackage{dcolumn}
\usepackage{bm}
\usepackage{hyperref}
\hypersetup{colorlinks,linkcolor=blue,anchorcolor=blue,citecolor=blue,urlcolor=blue}
\usepackage{color}
\usepackage{dcolumn} 
\usepackage{bm}      
\usepackage{array}   
\usepackage{amssymb}
\usepackage{amsmath}
\usepackage{amsfonts}
\usepackage{hyperref}
\usepackage{url}
\usepackage[T1]{fontenc}
\usepackage{multirow}
\usepackage{booktabs}
\usepackage[british]{babel}
\usepackage{times}
\usepackage{chemformula}
\usepackage{tikz} 
\usepackage{mathtools}
\usepackage{subcaption}
\usepackage{enumerate}
\usepackage{lipsum}
\usepackage{caption}

\newcommand{\ket}[1]{\left | #1 \right \rangle}
\newcommand{\bra}[1]{\left \langle #1 \right |}

\newtheorem{theorem}{Theorem}

\newtheorem{lemma}[theorem]{Lemma}

\usepackage{braket}

\begin{document}

\preprint{APS/123-QED}

\title{Non-Orthogonal Amplitude Amplification for Hybrid CV-DV Quantum Processors}

\author{Jiyuan Liu}
\affiliation{Department of Electrical and Computer Engineering, North Carolina State University, Raleigh, NC 27606, USA}
\author{Joel Bierman}
\email{jhbierma@ncsu.edu}
\affiliation{Department of Electrical and Computer Engineering, North Carolina State University, Raleigh, NC 27606, USA}
\author{Yuan Liu}
 \email{q_yuanliu@ncsu.edu}
\affiliation{Department of Electrical and Computer Engineering, North Carolina State University, Raleigh, NC 27606, USA}
\affiliation{Department of Computer Science, North Carolina State University, Raleigh, NC 27606, USA}
\affiliation{Department of Physics, North Carolina State University, Raleigh, NC 27607, USA}

\begin{abstract}
        We introduce a non-orthogonal amplitude amplification (NOAA) method for hybrid continuous-variable-discrete-variable (CV-DV) quantum processors, where the available CV-assisted projector is not perfectly selective for the target state. We establish a bound on the achievable target-state fidelity, which is determined by the initial target-state amplitude and the selectivity of the accessible reflection operator. This bound limits the performance of NOAA in hybrid CV-DV systems. To mitigate this limitation, we propose using a squeezed CV ancilla to reduce the overlap between relevant CV states and a hybrid quantum signal processing (HQSP) filter to suppress errors arising from imperfect knowledge of the target eigenenergy. Numerical simulations show that, for small initial target-state amplitudes, our protocol achieves higher fidelity and efficiency than the repeat-until-success (RUS) method despite the presence of non-orthogonality. For large initial amplitudes, RUS is already efficient, and amplitude amplification provides little additional advantage.
\end{abstract}

\maketitle

\section{Introduction}\label{sec:Intro}
Recent studies of hybrid oscillator-qubit quantum processors have demonstrated a powerful way to simulate quantum many-body systems by harnessing the infinite Hilbert space dimension of oscillators \cite{liuHybridOscillatorqubitQuantum2026,kemperHybridContinuousdiscretevariableQuantum2025a,nourseUsingBosonsImprove2025a,ganProvablyEfficientQuantum2025a}. New quantum algorithms for hybrid continuous-variable-discrete-variable (CV-DV) quantum processors and extensions of familiar qubit-based techniques such as block-encoding and quantum signal processing (QSP) \cite{martynGrandUnificationQuantum2021} are discussed in some works \cite{craneHybridOscillatorqubitQuantum2024,liuHybridOscillatorqubitQuantum2026}. In such protocols, the target non-unitary operation is implemented probabilistically. After an underlying CV-DV unitary evolution, it is obtained only when the bosonic mode yields a designated measurement outcome, such as a vacuum projection or a coherent-state POVM result \cite{bellCodesigningSpectralTransformation2025}. This repeat-until-success method (RUS) may have a heavy cost and small final fidelity with small initial amplitude of the target state and thus limit its application for efficiently preparing the target state with required fidelity. 

In qubit systems, a well-known technique for increasing the success probability of a probabilistic procedure is amplitude amplification (AA) \cite{martynGrandUnificationQuantum2021, brassardQuantumAmplitudeAmplification2002}, which generalizes Grover's search algorithm \cite{groverFastQuantumMechanical1996}. Grover's algorithm provides a quadratic speedup for unstructured search, but standard Grover iterations can suffer from the so-called souffl\'e problem, where repeated reflections about the target state and the initial state eventually rotate the state away from the target state. Fixed-point variants, such as the $\pi/3$ algorithm \cite{groverFixedpointQuantumSearch2005}, avoid this overshooting behavior at the cost of losing the quadratic speedup. Yoder \textit{et al.} subsequently introduced an optimal fixed-point search algorithm that avoids the souffl\'e problem while retaining the quadratic speedup within the more general framework of amplitude amplification \cite{yoderFixedPointQuantumSearch2014}. 

Performing AA in hybrid CV-DV quantum processors introduces an additional challenge. In AA, we need an implementable phase operation marking the desired success subspace. However, in certain problems, e.g. in state-preparation problems for hybrid CV-DV systems, the target state resides in the discrete-variable (DV) subsystem \cite{bellCodesigningSpectralTransformation2025} and is generally not known explicitly, making it difficult to directly construct a unitary that prepares the desired state. Thus, a selective projector onto the target DV state, and hence the corresponding reflection operator, cannot be constructed easily in practice. Alternatively, the reflection operation can be implemented through the continuous-variable (CV) subsystem. But this would bring new problems for implementing AA in the hybrid CV-DV systems. The CV states associated with different components of the DV system are generally non-orthogonal. Consequently, such a reflection does not perfectly isolate the target state and may also act on other components with nonzero overlap.

A good example is the protocol presented in Ref. \cite{bellCodesigningSpectralTransformation2025}, where Bell \textit{et al.} introduced a family of quantum eigenvalue-transformation oracles for hybrid CV-DV circuits that can be used to prepare eigenstates of Hamiltonians describing physical systems such as molecular and spin models. In this construction, the physical system is encoded in the DV subsystem, while the CV subsystem carries information associated with its eigenstates. The target eigenstate is generally not known explicitly, since preparing this state is itself the objective of the protocol. Therefore, a reflection using a projector that can unambiguously distinguish the target state from other states is generally unavailable in practice, whereas a CV-assisted reflection can still be implemented.  Unlike an ideal target-state reflection, the resulting CV-assisted reflection is not perfectly selective because the oscillator states associated with different DV eigenstates are generally non-orthogonal.

Previous studies have considered several forms of imperfect amplitude amplification. H\o{}yer studied AA with arbitrary phases and its tolerance to phase errors, while Tulsi developed quantum search algorithms robust to reproducible and reversible errors in the selective transformations \cite{Hoyer2000,Tulsi2008}. Approximate implementations of a target-state reflection have also been studied, with bounds on the resources required to achieve a prescribed reflection accuracy and on the accumulation of such errors during fixed-point AA \cite{Chowdhury2018,Orsucci2018}, while stochastic oracle faults can eliminate the quadratic search speedup altogether \cite{RegevSchiff2008}. Related work on oblivious amplitude amplification has instead quantified errors arising when the block-encoded transformation is only approximately unitary or genuinely nonunitary \cite{Berry2015,Zecchi2025}. Ref.
\cite{Zecchi2025} further derives a lower bound on the resulting state fidelity in terms of the degree of nonunitarity. These results quantify limitations arising from phase errors, approximate reflections, faulty oracles, or nonunitary transformations. Here, we instead consider amplitude amplification with a fixed accessible projector that is implemented exactly but does not isolate the desired component of the encoded state. In hybrid CV-DV systems, this limitation arises because different DV eigenstates are correlated with nonorthogonal CV states. We refer to amplitude amplification based on such a nonselective reflection arising from non-orthogonal CV states as non-orthogonal amplitude amplification (NOAA). 

In this work, we show that the achievable fidelity of NOAA is bounded by a ratio dependent on the initial amplitude of the target state and the selectivity of the accessible projector. This bound cannot be removed by increasing the amplification degree or optimizing the AA phases and limits the performance of NOAA in hybrid CV-DV systems. We further characterize the additional fidelity reduction caused by target-energy mismatch.

To mitigate these limitations, we introduce a hybrid quantum signal processing (HQSP) filter, motivated by recent developments in quantum signal processing for hybrid CV-DV architectures \cite{rossiQuantumSignalProcessing2023, liuMixedAnalogdigitalQuantum2025a}. Our numerical simulations show that the HQSP filter can substantially improve NOAA fidelity by suppressing contributions associated with imperfect state discrimination and improving the robustness of NOAA against errors in energy estimation.

This work is structured as follows. In Sec.~\ref{sec:background}, we introduce the background of the NOAA method. In Sec.~\ref{sec:NOAA}, we analyze the NOAA fidelity and derive an upper bound on the achievable fidelity. In Sec.~\ref{sec:HQSP}, we introduce an HQSP filter to mitigate the effect of non-orthogonality and the mismatch effect to improve the NOAA fidelity. Finally, we will conclude our work in Sec.~\ref{sec:conclusion}.

\section{Background}\label{sec:background}

In this section, we introduce the main ingredients required for non-orthogonal amplitude amplification (NOAA). We first review the hybrid CV-DV spectral-filtering oracle of Ref.~\cite{bellCodesigningSpectralTransformation2025} and the corresponding repeat-until-success (RUS) state-preparation protocol in Sec. \ref{ssec:filter}. We then summarize standard amplitude amplification (AA), emphasizing the target-state reflection required by the protocol in Sec. \ref{ssec:aa}. Finally, we briefly review single-variable hybrid quantum signal processing (HQSP) in Sec. \ref{ssec:hqsp}, which will later be used to improve the reflection implemented in the CV subsystem.

\subsection{Hybrid CV-DV spectral filtering}
\label{ssec:filter}
We first consider the spectral-transformation oracle introduced by Bell \textit{et al.}~\cite{bellCodesigningSpectralTransformation2025}. Let $\hat{H}$ be a Hamiltonian acting on the DV subsystem, with eigenstates $\ket{E_n}$ and corresponding eigenenergies $E_n$, ordered such that $E_0<E_1<\cdots$. We consider an input state 
\begin{equation} 
\ket{\Psi_{\mathrm{in}}} = \sum_n \psi_n \ket{E_n}_q\otimes\ket{0}_o,
\label{eq:inputStateBackground}
\end{equation} 
where the oscillator ancilla is initialized in the vacuum state.

The key operation in Ref.~\cite{bellCodesigningSpectralTransformation2025} is an energy-dependent displacement of the oscillator. For convenience, we denote the corresponding hybrid CV-DV oracle by $P_{\hat{H}}$. Its action can be written as 
\begin{equation} 
\ket{\Psi_0} = P_{\hat{H}}\ket{\Psi_{\mathrm{in}}} 
= \sum_n \psi_n \ket{E_n}_q\otimes\ket{\alpha_n}_o, \label{eq:iniStateBackground} 
\end{equation} 
where $\ket{\alpha_n}_o$ is a coherent state with displacement 
\begin{equation} \alpha_n=\alpha(E_n-E_s). \label{eq:alpha_n_background}
\end{equation} 
Here, $\alpha$ determines the strength of the displacement, while $E_s$ is an energy shift chosen to select the desired part of the spectrum. Therefore, different energy components of the DV state are correlated with different coherent states of the oscillator.

The non-unitary spectral filter is obtained by projecting the oscillator onto the vacuum state with the projector $\hat{P}_{vac} = \hat{I}_{q}\otimes\ket{0}_o\bra{0}_o$. The DV state after a successful vacuum measurement is 
\begin{equation} 
\ket{\Psi_{\mathrm{RUS}}} = \frac{1}{\sqrt{P_s}} \sum_n \psi_n e^{-|\alpha_n|^2/2} \ket{E_n}_q, \label{eq:RUS_state}
\end{equation} 
where the success probability is
\begin{equation} 
P_s = \sum_n |\psi_n|^2 e^{-|\alpha_n|^2}. \label{eq:RUS_prob}
\end{equation} 
We will present a more detailed discussion of RUS later.
With the expression of the final state, we can see that the vacuum postselection implements a Gaussian spectral filter centered at $E_s$.

For concreteness, throughout this work we take a nondegenerate ground state $\ket{E_0}$ as the target state. The same discussion can be extended to an isolated excited state by appropriately choosing $E_s$. In the ideal case, $E_s=E_0$, such that the target component remains centered at the oscillator vacuum, while the excited-state components are displaced by amounts proportional to their energy gaps. Increasing $\alpha$ therefore enhances the spectral selectivity by suppressing the excited-state contributions. This improvement, however, comes with important limitations. First, stronger filtering does not increase the initial weight of the target state: in the large-$\alpha$ limit, the RUS probability approaches $|\psi_0|^2$, so a small initial target-state amplitude still leads to a large repetition cost. Second, for a compiled implementation of the spectral oracle, increasing $\alpha$ can also increase the approximation error and eventually lead to eigenstate hybridization \cite{bellCodesigningSpectralTransformation2025}. Therefore, the displacement strength must balance spectral selectivity against postselection cost and implementation error.

In addition to the limitation mentioned above, in practice, the target eigenenergy may only be known approximately. 
Let $\tilde{E}_t$ denote an estimate of the target eigenenergy $E_t$ and define
\begin{equation}
  \delta\tilde{E}_t = E_t - \tilde{E}_t .\label{eq:displacementMismatch}
\end{equation}
Since we take the ground state as the target throughout this work, we set $t=0$ in the following. Choosing $E_s=\tilde{E}_0$, Eq.~\eqref{eq:alpha_n_background} becomes $\alpha_n = \alpha(E_n-E_0)+\delta\alpha_0$, with $\delta\alpha_0=\alpha\delta\tilde{E}_0$. The target component is therefore displaced from the vacuum by $\delta\alpha_0$, which suppresses its vacuum-transition amplitude and can degrade the selectivity of the RUS filter. We will analyze its influence in Sec.~\ref{sec:NOAA}.

The RUS protocol motivates the use of amplitude amplification to improve the success probability. However, the hybrid construction introduces an additional difficulty. The coherent states associated with different energy components are not mutually orthogonal. As a result, a reflection implemented solely through the oscillator generally cannot act as an exact reflection about a single DV eigenstate. This distinction from standard AA is the central motivation for NOAA.

\subsection{Amplitude amplification} \label{ssec:aa}
Amplitude amplification (AA) provides a systematic framework for amplifying the target-subspace component of an initial state
\cite{brassardQuantumAmplitudeAmplification2002,
yoderFixedPointQuantumSearch2014}.
Let the initial state be decomposed as
\begin{equation}
    \ket{\Psi_0}=\psi_0\ket{\Psi_t}+\psi_1\ket{\Psi_\perp},
    \label{eq:AADecomposition}
\end{equation}
where $\ket{\Psi_t}$ is the target state and $\ket{\Psi_\perp}$ is the normalized component in its orthogonal complement. The coefficients satisfy $|\psi_0|^2+|\psi_1|^2=1$.

In the generalized AA framework, the amplification procedure is constructed from generalized reflection operators\cite{yoderFixedPointQuantumSearch2014}. These reflection operators are defined as $\hat{S}_{s}(\alpha)=e^{-i\alpha \ket{\Psi_0}\bra{\Psi_0}}
$ and $\hat{S}_{t}(\beta)=e^{i\beta \hat{\Pi}_{t}}$, where $\hat{\Pi}_{t}$ is the projector onto the target subspace. The generalized Grover iterate is defined as
\begin{equation}
    \hat{G}(\alpha,\beta)=-\hat{S}_{s}(\alpha)\hat{S}_{t}(\beta),
\end{equation}
which acts as an $SU(2)$ rotation within the two-dimensional subspace spanned by $\ket{\Psi_t}$ and $\ket{\Psi_\perp}$
\cite{groverFastQuantumMechanical1996,
yoderFixedPointQuantumSearch2014}.

In the fixed-point amplitude amplification framework of Yoder \textit{et al.}, a general AA protocol is constructed by composing a sequence of such generalized Grover iterates with appropriately chosen phase angles,
\begin{equation}
    \hat{U}_{L}=\prod_{j=1}^{l}\hat{G}(\alpha_j,\beta_j),\label{eq:GroverIteratesSequence}
\end{equation}
where $L=2l+1$, and the corresponding query complexity is
$L-1=2l$ \cite{yoderFixedPointQuantumSearch2014}.
The phase schedule $\{\alpha_j,\beta_j\}$ determines the amplification behavior. In particular, when all phases are chosen as $
    \alpha_j=\beta_j=\pm\pi
$, each generalized Grover iterate reduces to the standard Grover iterate, recovering the usual quadratic speedup. More general choices of the phase angles allow one to construct fixed-point AA protocols that avoid the overshooting (souffl\'e) problem while retaining the quadratic speedup
\cite{yoderFixedPointQuantumSearch2014}.

\subsection{Single-variable hybrid quantum signal processing} 
\label{ssec:hqsp}
We finally review the single-variable HQSP construction used later to design a CV filter. Hybrid QSP extends the polynomial-transformation framework of QSP to hybrid CV-DV systems \cite{liuMixedAnalogdigitalQuantum2025a}. For one qubit coupled to an oscillator, the basic signal operator is the conditional displacement 
\begin{equation} W_z^{(\kappa)} = e^{-i\frac{\kappa}{2}\hat{x}\sigma_z} = \begin{bmatrix} 
\hat{w} & 0 \\ 0 & \hat{w}^{-1} 
\end{bmatrix}, \qquad \hat{w} = e^{-i\frac{\kappa}{2}\hat{x}}, \label{eq:HQSPWBackground}
\end{equation} 
where $\kappa$ controls the displacement strength. A degree-$d$ single-variable HQSP sequence is constructed by interleaving $W_z^{(\kappa)}$ with parameterized qubit rotations, 
\begin{align} 
U_{\mathrm{HQSP}}(\hat{w}) 
&= e^{i\theta_0\sigma_x} \prod_{j=1}^{d} W_z^{(\kappa)}e^{i\theta_j\sigma_x} \nonumber\\ 
&= \begin{bmatrix} F(\hat{w}) & iG(\hat{w}) \\ 
iG(\hat{w}^{-1}) & F(\hat{w}^{-1})
\end{bmatrix}. 
\label{eq:HQSPBackground} 
\end{align} 
Here, $F(\hat{w})$ and $G(\hat{w})$ are degree-$d$ Laurent polynomials whose coefficients are determined by the phase angles $\{\theta_j\}$. Unitarity requires \begin{equation} F(\hat{w})F(\hat{w}^{-1}) + G(\hat{w})G(\hat{w}^{-1}) = I. \label{eq:HQSPUnitaryBackground} \end{equation} By appropriately choosing the phase angles, $F(\hat{w})$ can be designed to approximate a desired bounded function of the oscillator quadrature \cite{liuMixedAnalogdigitalQuantum2025a, dongEfficientPhasefactorEvaluation2021}. In the present work, we use this single-variable construction to design an HQSP filter that improves the selectivity of the CV-assisted reflection. The specific polynomial and its application to NOAA are introduced in Sec.~\ref{sec:HQSP}.

\section{Non-orthogonal amplitude amplification on hybrid CV-DV quantum processors}
\label{sec:NOAA}

In this section, we first introduce the problem setting of non-orthogonal amplitude amplification (NOAA) in Sec. \ref{ssec:setup}. We then analyze the NOAA fidelity and derive its upper bound in Sec. \ref{ssec:theory-noaa}. Applications in hybrid CV-DV setting are discussed in Sec. \ref{ssec:cv-dv-app}.

\subsection{Problem setup} \label{ssec:setup}

We consider the initial state introduced in Eq.~\eqref{eq:iniStateBackground}. To isolate the effect of the non-orthogonality of the CV states, we first assume $\delta\tilde{E}_0=0$. In this case, the CV state associated with the target eigenstate is centered at the vacuum, i.e., $\alpha_0=0$. We therefore construct the target projector as
\begin{align}
    \hat{\Pi}=\hat{I}_q\otimes\ket{0}_o\bra{0}_o .\label{eq:PIbackground}
\end{align}
For the target component, this projector satisfies
$\hat{\Pi}\ket{E_0}_q\ket{0}_o=\ket{E_0}_q\ket{0}_o $. For an undesired component, however, the projected state is generally nonzero because the CV states $\ket{\alpha_n}_o$ are not orthogonal to the vacuum. Hence, although $\hat{\Pi}$ preserves the target CV state, it also has a nonzero projection onto the CV states
associated with undesired eigenstates. Together with the initial-state projector $\hat{\Pi}_s=\ket{\Psi_0}\bra{\Psi_0}$, they define the reflection operators used in the Grover sequence introduced in Sec.~\ref{sec:background}. 
The resulting sequence defines the NOAA protocol. 

We now consider the case $\delta\tilde{E}_0\neq0$, corresponding to an imperfect estimate of the target eigenenergy. We refer to the resulting degradation as the \textit{mismatch effect}. The experimentally implemented projector is still chosen as
\begin{align}
    \hat{\Pi}^{\mathrm{mis}}=\hat{I}_q\otimes\ket{0}_o\bra{0}_o .
\end{align}
However, the CV state associated with the target eigenstate is no longer the vacuum. Instead,
$\alpha_0=\delta\alpha_0=\alpha\,\delta\tilde{E}_0\neq 0$ and 
$\hat{\Pi}^{\mathrm{mis}}\ket{E_0}_q\ket{\alpha_0}_o=\braket{0|\alpha_0}\ket{E_0}_q\ket{0}_o\neq\ket{E_0}_q\ket{\alpha_0}_o$. Thus, unlike the matched case, the available projector no longer preserves the target component itself.

As shown below, the mismatch introduces an additional reduction of the NOAA fidelity bound, characterized by a factor that depends on the displacement mismatch and the corresponding overlap of the target CV state with the projector.

The above discussion identifies two distinct effects that limit the fidelity of NOAA. First, even in the absence of eigenenergy mismatch, the CV projector has nonzero overlaps with undesired components because the corresponding CV states are not orthogonal to the target CV state. Second, when $\delta\tilde{E}_0\neq0$, the target CV state itself is no longer perfectly
matched to the available projector.

To mitigate these effects, in Sec.~\ref{sec:HQSP} we introduce a squeezed CV ancilla and a hybrid quantum signal processing (HQSP) filter \cite{liuMixedAnalogdigitalQuantum2025a}. The squeezing reduces the overlap between CV states associated with different eigenvalues, while the HQSP filter provides additional spectral selectivity and improves the robustness against eigenenergy mismatch. We compare the resulting protocol with repeat-until-success (RUS) postselection numerically and show that, for a small initial amplitude of the target eigenstate, the proposed protocol can
achieve higher fidelity with a lower preparation cost than RUS.

\subsection{Theory of NOAA} \label{ssec:theory-noaa}
In this subsection, we introduce a general description of NOAA based on the setup discussed above. We consider an initial state expanded in an orthonormal basis $\{\ket{\xi_n}\}$, 
\begin{equation} 
\hat{U}_{\text{prep}}\ket{0}_q\ket{0}_o = \ket{\Psi_0} = \sum_n \psi_n\ket{\xi_n}, \label{eq:initial_state} 
\end{equation} 
where $\hat{U}_{\text{prep}}$ is the state-preparation circuit. We take $\ket{\xi_0}$ as the target state. The goal of NOAA is to increase the fidelity of $\ket{\xi_0}$ in the final state. In the ideal case, we have access to the target projector 
\begin{equation} 
\hat{\Pi}_0 = \ket{\xi_0}\bra{\xi_0}. 
\end{equation} 
Standard amplitude amplification can then be used to amplify the target component \cite{yoderFixedPointQuantumSearch2014}.  Instead, NOAA uses an accessible projector $\hat{\Pi}_{\mathrm{NOAA}}$. For a given projector $\hat{\Pi}_{\mathrm{NOAA}}$, the corresponding target-selective phase operator is $\hat{S}_t(\beta)=e^{i\beta\hat{\Pi}_{\mathrm{NOAA}}}$. As discussed in Sec.~\ref{ssec:setup}, we consider two different cases. Each case provides a specific $\hat{\Pi}_{\text{NOAA}}$. The first case is the matched case with $\delta\tilde{E}_0=0$, and the second is the mismatched case with $\delta\tilde{E}_0\neq0$.

We first describe these two cases using an explicit basis. In the matched case, the accessible projector contains the target state together with a
leakage subspace. We write
\begin{equation}
    \hat{\Pi}
    =\ket{\xi_0}\bra{\xi_0}+\sum_k \ket{\xi_{\epsilon,k}} \bra{\xi_{\epsilon,k}},
    \label{eq:matched_projector}
\end{equation}
where $\{\ket{\xi_{\epsilon,k}}\}$ is an orthonormal basis of the leakage subspace and $\braket{\xi_0|\xi_{\epsilon,k}}=0$. Thus, $\hat{\Pi}$ contains the target state but does not project onto the
target state alone.

In the mismatched case, the subspace selected by the accessible projector
does not contain the exact target state. We write
\begin{equation}
    \hat{\Pi}^{\mathrm{mis}}
    =\sum_k \ket{\tilde{\xi}_k} \bra{\tilde{\xi}_k},
    \label{eq:mismatched_projector}
\end{equation}
where $\{\ket{\tilde{\xi}_k}\}$ is an orthonormal basis of the mismatched
subspace. From this definition, we have $\hat{\Pi}^{\mathrm{mis}}\ket{\xi_0}\neq\ket{\xi_0}$.

To identify the part of the mismatched subspace associated with the target
state, we define
\begin{equation}
\ket{\xi_0^{\mathrm{mis}}}=\frac{\hat{\Pi}^{\mathrm{mis}}\ket{\xi_0}}{\sqrt{\bra{\xi_0}\hat{\Pi}^{\mathrm{mis}}\ket{\xi_0}}},
\end{equation}
where $ 0<\bra{\xi_0}\hat{\Pi}^{\mathrm{mis}}\ket{\xi_0}<1$ and the corresponding helper projector
\begin{equation}\hat{\Pi}_0^{\mathrm{mis}}=\ket{\xi_0^{\mathrm{mis}}}\bra{\xi_0^{\mathrm{mis}}}.
\end{equation}

We can now define the two cases without referring to a particular basis. For the matched case, the target subspace is contained in the subspace defined by $\hat{\Pi}$. Therefore, 
\begin{align} 
\hat{\Pi}\hat{\Pi}_0 = \hat{\Pi}_0\hat{\Pi} = \hat{\Pi}_0. \label{eq:unmismatch_relation} 
\end{align}
For the mismatched case, the target subspace is not contained in the subspace defined by $\hat{\Pi}^{\mathrm{mis}}$. Therefore, 
\begin{align} 
\hat{\Pi}^{\mathrm{mis}}\hat{\Pi}_0 \neq \hat{\Pi}_0,\,\, 
\hat{\Pi}_0\hat{\Pi}^{\mathrm{mis}}\neq \hat{\Pi}_0,
\label{eq:mismatch_relation} 
\end{align} 
We will use the relation $\hat{\Pi}^{\mathrm{mis}}\hat{\Pi}_0^{\mathrm{mis}}=\hat{\Pi}_0^{\mathrm{mis}}\hat{\Pi}^{\mathrm{mis}}=\hat{\Pi}_0^{\mathrm{mis}}$, which can be interpreted from the definition of $\hat{\Pi}_0^{\text{mis}}$.
These relations provide the general assumptions used in the following analysis. For the matched case, we set $\hat{\Pi}_{\mathrm{NOAA}}=\hat{\Pi}$, while for the mismatched case we set $\hat{\Pi}_{\mathrm{NOAA}}=\hat{\Pi}^{\mathrm{mis}}$.

Let us first focus on the former situation, i.e.  $\hat{\Pi}_{\text{NOAA}} =\hat{\Pi}$ is used for the Grover iterates sequence in Eq.~\eqref{eq:GroverIteratesSequence}. The NOAA fidelity of matched case is given in lemma~\ref{lem:Fidelity_of_NOAA}.
\begin{lemma}
[Matched NOAA fidelity]
\label{lem:Fidelity_of_NOAA}
    The final state after NOAA of degree $l$ lies in the subspace spanned by $\{\ket{\Psi_0},\hat{\Pi}\ket{\Psi_0}\}$, and can be written as
    \begin{align}
        \ket{\Psi_{f}} = \hat{U}_{L}\ket{\Psi_{0}} = \gamma_{l}\ket{\Psi_{0}}-\theta_{l}\hat{\Pi} \ket{\Psi_{0}}.
        \label{eq:Psi_f}
    \end{align}
    $\gamma_{l}$ and $\theta_{l}$ are some complex coefficients and $L=2l+1$. The fidelity of $\ket{\Psi_{f}}$ is given by 
    \begin{equation}
        \mathcal{F}_{\text{NOAA}} = \braket{\Psi_f|\hat{\Pi}_{0}|\Psi_{f}} = |\gamma_{l}-\theta_{l} |^{2} \braket{ \Psi_{0}|\hat{\Pi}_{0} | \Psi_{0} }  \label{eq:Fidelity_of_NOAA}
    \end{equation}
    where $\hat{\Pi}_{0}$ is defined as the ideal projector onto the ideal target subspace, and satisfies $\hat{\Pi}_{0} \hat{\Pi} = \hat{\Pi} \hat{\Pi}_{0}= \hat{\Pi}_{0}$.
\end{lemma}
\textit{Proof}. The proof of Lemma \ref{lem:Fidelity_of_NOAA} is given in Appendix \ref{apendix:Lemma1}.

To better understand Eq.~\eqref{eq:Fidelity_of_NOAA}, we introduce the notation $P_{L}$ to evaluate the performance of AA for amplifying the fidelity of the states in the subspace defined by the projector $\hat{\Pi}$:
\begin{equation}
    P_{L} = \braket{\Psi_{f}|\hat{\Pi}|\Psi_{f}}=|\gamma_{l} - \theta_{l}|^{2} \braket{\Psi_{0}|\hat{\Pi}|\Psi_{0}}. \label{eq:PL}
\end{equation}
Combining Eq.~\eqref{eq:Fidelity_of_NOAA} and Eq.~\eqref{eq:PL}, the NOAA fidelity can be expressed with $P_L$:
\begin{equation}
    \mathcal{F}_{\text{NOAA}} = \frac{\braket{\Psi_{0}|\hat{\Pi}_{0}|\Psi_{0}}}{\braket{\Psi_{0}|\hat{\Pi}|\Psi_{0}}} P_{L}. \label{eq:Fidelity_o f_NOAA_bound}
\end{equation}
Equation~\eqref{eq:Fidelity_o f_NOAA_bound} separates the performance of NOAA into two factors. The quantity $P_L$ describes how efficiently amplitude amplification transfers population into the subspace selected by $\hat{\Pi}$, whereas $\eta=
\langle\Psi_0|\hat{\Pi}_0|\Psi_0\rangle/\langle\Psi_0|\hat{\Pi}|\Psi_0\rangle$
is the fraction of the initially projected amplitude associated with the desired target state. Therefore, optimizing the amplification sequence can increase $P_L$ toward unity but cannot modify $\eta$. For a fixed accessible projector, the maximum achievable fidelity is thus $F_{\mathrm{NOAA}}\leq\eta$. Improving this ceiling requires increasing the selectivity of the accessible projector rather than increasing the amplification degree.

We next consider the mismatched case introduced in Sec.~\ref{ssec:setup}. Because the target eigenenergy is known only approximately, the accessible projector no longer contains the exact target hybrid state. We therefore set $\hat{\Pi}_{\text{NOAA}}=\hat{\Pi}^{\text{mis}}$. The resulting target-state fidelity is characterized by Lemma \ref{lem:fidelityMismatch}.

\begin{lemma}[Fidelity of mismatched NOAA]\label{lem:fidelityMismatch} 
The fidelity of mismatched NOAA is defined as
\begin{equation}
\mathcal{F}^{\textrm{mis}}_{\text{NOAA}} = \braket{\tilde{\Psi}_{f}|\hat{\Pi}_{0}|\tilde{\Psi}_{f}},
\end{equation}
We define
\begin{equation}
\mathcal{K}\equiv\braket{\tilde{\Psi}_{f}|\hat{\Pi}^{\textrm{mis}}_0\hat{\Pi}_{0}\hat{\Pi}^{\textrm{mis}}_0|\tilde{\Psi}_{f}} =||\hat{\Pi}_{0}\hat{\Pi}^{\textrm{mis}}_0\ket{\tilde{\Psi}_{f}}||^{2}\label{eq:def_of_mathcalK}
\end{equation}
as the quantity representing the overlap of the mismatched target state defined by $\hat{\Pi}^{\text{mis}}$ and the target state $\ket{\xi_0}$.
Define a final state dependent mismatch error:
\begin{equation}
    \Delta_{f} = ||\hat{\Pi}_{0}(\hat{I} - \hat{\Pi}^{\textrm{mis}}_0)\ket{\tilde{\Psi}_{f}}||.\label{eq:lem2error}
\end{equation} 
The NOAA fidelity with $\hat{\Pi}^{\text{mis}}$ is therefore bounded by
\begin{equation}
    (\sqrt{\mathcal{K}} - \Delta_{f})^2 \le \mathcal{F}^{\textrm{mis}}_{\text{NOAA}} \le (\sqrt{\mathcal{K}} + \Delta_{f})^2.\label{eq:Fidelity_of_mismatched_NOAA_bound}
\end{equation}
\end{lemma}
\textit{Proof}. The proof of Lemma 2 is given in Appendix \ref{apendix:Lemma2}.

In the next subsection, we will analyze the mismatched case in hybrid CV-DV system, from which we will see the fidelity of mismatched NOAA is also bounded by a ratio defined by initial state amplitude. 

Now, let us introduce the phase angles used in NOAA sequence.
In qubit systems, AA can achieve the desired success probability with phases given in Ref. \cite{yoderFixedPointQuantumSearch2014}:
\begin{equation}
    \alpha_{j}=2\cot ^{-1}\left( \tan\left( \frac{2\pi j}{L} \right)\sqrt{ 1-\gamma^{2} } \right)=-\beta_{l-j+1},
    \label{eq:YoderPhase}
\end{equation}
where $\gamma^{-1}=T_{1/L}(1/\delta)$. The fidelity after the fixed point AA is
\begin{equation}
    P_{L}^{\text{FP}}(\lambda)=1-\delta^{2}T_{L}(T_{1/L}(1 /\delta)\sqrt{ 1-\lambda})^{2}, \label{eq:OAAsPL}
\end{equation}
where $\lambda$ is the modulus squared of the target state's coefficient in the initial state, $T_{L}(x)$ is the Chebyshev polynomial of the first kind and $\delta$ is a parameter that controls the error of AA.
In the non-orthogonal case, we can use these phase angles to optimize the $P_{L}$ defined in Eq.~\eqref{eq:PL}. For simplicity, we will denote $P_L^{\text{FP}}$ as $P_L$ in the following analysis. Analogous to AA, NOAA with these phase angles shows a quasi fixed-point feature:
\begin{lemma}
[NOAA's fixed point property]\label{lem:FPNOAA}
In this lemma, we consider NOAA with $\delta\tilde{E}_0 =0$.
Using the phase angles defined in Eq.~\eqref{eq:YoderPhase}, after the NOAA, the fidelity is 
\begin{equation}
  \mathcal{F}_{\text{NOAA}}=\frac{\braket{ \Psi_{0}|\hat{\Pi}_{0} |\Psi_{0}  } }{\braket{ \Psi_{0}|\hat{\Pi} |\Psi_{0}  } }P_L(\braket{ \Psi_{0}|\hat{\Pi} |\Psi_{0}  }),
\end{equation}
where $P_L$ is given in Eq.~\eqref{eq:OAAsPL}.

If  $\braket{ \Psi_{0}|\hat{\Pi} |\Psi_{0}  }  \ge w = 1-T_{1 /L}(1 /\delta)^{-2}\approx (\log(2 /\delta)/L)^2$, we have 
\begin{equation}
    \eta(1-\delta^{2})\le \mathcal{F}_{\text{NOAA}} \le \eta, \label{eq:bounds_inequ}
\end{equation} where $\eta=\langle \Psi_{0}\mid\hat{\Pi}_{0} \mid\Psi_{0}  \rangle /\langle \Psi_{0}\mid\hat{\Pi} \mid\Psi_{0}   \rangle$.
\end{lemma}
\textit{Proof}. The proof of Lemma \ref{lem:FPNOAA} is given in Appendix \ref{apendix:Lemma3}.

In the rest of the manuscript, we will use these phase angles for analysis and numerical simulation of the NOAA fidelity in the hybrid CV-DV system.

\subsection{NOAA on Hybrid CV-DV system} \label{ssec:cv-dv-app}
In this section, we analyze the fidelity of the NOAA for hybrid CV-DV systems. In CV-DV systems, the initial state can be expressed by 
\begin{equation}
   \ket{\Psi_{0}} = \hat{U}_{\text{prep}}\ket{0}_{q}\ket{0}_{o} = \sum_{n} \psi_{n} \ket{E_{n}}_{q} \ket{\alpha_{n}}_{o}, \label{eq:CVDVinistate}
\end{equation}
where $\hat U_{\text{prep}}$ is the initial state preparation circuit, $\ket{E_{n}}_{q}$ are the eigenstates of the target Hamiltonian, $\ket{\alpha_{n}}_{o}$ are coherent states as the CV part and are proportional to eigenenergies $E_n$. Again, as we emphasized in previous sections, the DV part, representing the target physical system, is implicit to us and all we have is an estimation of the eigenenergy with error $\delta\tilde{E}_t$, defined in Eq.~\eqref{eq:displacementMismatch}. We assume $\ket{E_{0}}_q\ket{\alpha_0}_o$ is the target state for NOAA in the hybrid CV-DV systems. Therefore the ideal projector is
\begin{align}
    \hat{\Pi}_0=\ket{E_{0}}_q\bra{E_{0}}_q\otimes\ket{\alpha_0}_o\bra{\alpha_0}_o.\label{eq:CVDVidealProjector}
\end{align}
The analysis of other choices for target state will be exactly the same. 

For the matched situation, we have $\delta \tilde{E}_0=0$ and $\alpha_0=0$. We use the projector $\hat{\Pi}$ defined in Eq.~\eqref{eq:PIbackground} as the accessible projector. The NOAA fidelity is
\begin{align}
    \mathcal{F}_{\text{NOAA}} = \eta P_{L} 
                              =\frac{P_{L}}{\sum_{n} |\psi_{n}|^{2} |\braket{\alpha_{n}|\alpha_{0}}|^{2}/|\psi_{0}|^{2}}.\label{eq:CVDVfidelityOfMatchNOAA}
\end{align}
As we mentioned above, Eq.~\eqref{eq:CVDVfidelityOfMatchNOAA} contains two distinct contributions to NOAA fidelity. The factor $P_L$ characterizes the intrinsic performance of the fixed-point amplification sequence, whereas the denominator describes the residual amplitude of the undesired eigenstates after the coherent-state projection. To make this explicit, we define
\begin{align}
    \mathcal{R}=\sum_{n \neq 0} |\psi_{n}|^{2} |\braket{\alpha_{n}|\alpha_{0}}|^{2}/ |\psi_0|^{2}.\label{eq:mathcalR}
\end{align}
The fidelity can be written as $\mathcal{F}_{\text{NOAA}}=P_L/(1+\mathcal{R})$. Coherent states' non-orthogonality introduces a fidelity reduction that is independent of the optimization on the amplification sequence phase angles. In the weak-overlap regime, $\mathcal{R}\ll1$, $\mathcal{F}_{\text{NOAA}}=P_L\left[1-\mathcal{R}+\mathcal{O}(\mathcal{R}^2)\right]$. Thus, small individual coherent-state overlaps do not by themselves guarantee high fidelity. As quantified by $R$ in Eq.~\eqref{eq:mathcalR}, their cumulative weighted contribution must remain small relative to the initial target-state weight. Only in this regime is the overlap-induced reduction in fidelity negligible.

For simplicity, let us consider a simple two-state model, where the initial state is given 
\begin{align}
    \ket{\Psi_0^{(2)}} = \psi_0 \ket{E_0}_q\ket{\alpha_0}_o + \psi_1 \ket{E_1}_q\ket{\alpha_1}_o.\label{eq:2statesModel}
\end{align}
A more general model with more than two states can be easily inferred from this simple model. The fidelity for the two-state model is
\begin{align}
    \mathcal{F}_{\text{NOAA}} = \eta P_L(|\psi_0|,\epsilon),
    \label{eq:two_state_NOAA}
\end{align}
where $\epsilon=\braket{\alpha_1|\alpha_0}=\exp (-|\alpha_1-\alpha_0|^2 /2)$ and we emphasize the dependence on $\psi_0$ and $\epsilon$ of $P_L$. We assume $\{\alpha_n\}$ are real for the discussions in this section. 
In Eq.~\eqref{eq:two_state_NOAA}, the factor $\eta\le1$ represents a non-orthogonality-induced fidelity ceiling. In particular, even a small coherent-state overlap may substantially reduce the fidelity when the initial target amplitude $|\psi_0|^2$ is small. The desired $\eta \to1$ will be approached either in the orthogonal limit $\epsilon\to0$ or when the initial state already has a large overlap with the target state.  

\begin{figure}
\centering
\includegraphics[width=\columnwidth]{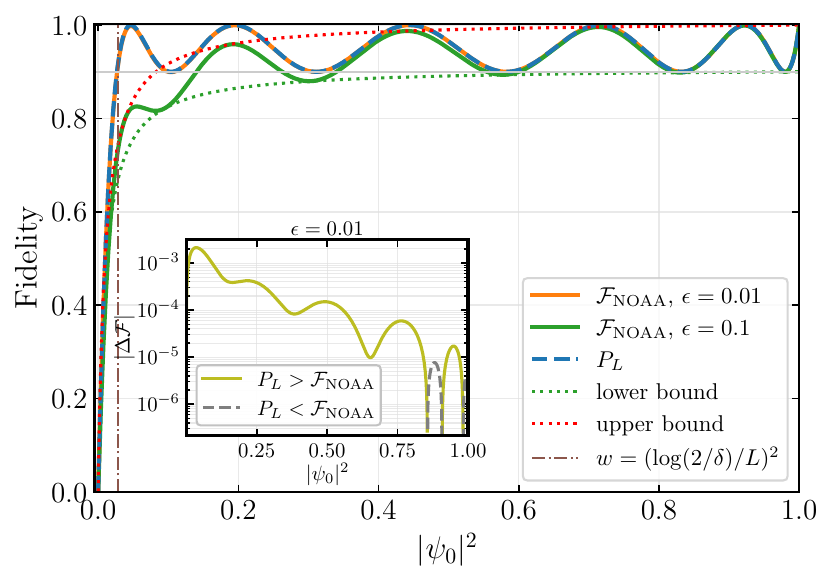}
\caption{\textbf{A comparison between AA and NOAA for two-state model.} Both NOAA sequence and AA sequence have degree $l=5$ and the error parameter $\delta^2=0.1$. The AA fidelity $P_L$ is given by Eq.~\eqref{eq:OAAsPL}. The NOAA fidelity $\mathcal{F}_{\text{NOAA}}$ is given by Eq.~\eqref{eq:two_state_NOAA}. The upper bound of NOAA (red dotted line) $\eta$ is defined in Lemma \ref{lem:FPNOAA} and the corresponding lower bound (green dotted line) is $\eta(1-\delta^2)$. The horizontal gray solid line represents $1-\delta^2$, the lower bound of AA. Inset shows the difference $|\Delta \mathcal{F}|=|P_L-\mathcal{F}_{\text{NOAA}}|$, where $\epsilon$ for NOAA here is $0.01$.}
\label{fig:2state}
\end{figure}

In Fig.~\ref{fig:2state}, we compare the NOAA fidelity with the AA fidelity (with $\epsilon=0$) and the upper and lower bound of NOAA defined in Eq.~\eqref{eq:bounds_inequ}. For $\epsilon=0.1$, the NOAA fidelity is lower than that of AA at small $|\psi_0|^2$ for the fixed-point phase angles considered here. More generally, the NOAA fidelity is bounded by $\eta$ independently of the choice of the phase angles $\{\alpha_j,\beta_j\}$. Therefore, although different phase choices may change the factor $P_L$, they cannot remove the non-orthogonality-induced fidelity ceiling. For $\epsilon=0.01$, NOAA's performance is very close to AA, the inset shows the difference of the fidelity between NOAA of $\epsilon=0.01$ and AA. As shown in the numerical results, NOAA also shows certain fixed point feature for small non-orthogonality, for $|\psi_0|^2>(w-\epsilon^2)/(1-\epsilon^2)\approx w$, the NOAA fidelity is above $\eta(1-\delta^2)$, where $w$ is defined in Lemma \ref{lem:FPNOAA}.

For the mismatched situation where $\delta \tilde{E}_0 \neq 0$, the value of $\alpha_0$ becomes $\alpha\delta \tilde{E}_0=\delta \alpha_0$. The mismatched projector is expressed as $\hat{\Pi}^{\textrm{mis}} = \hat{I}_q\otimes \ket{0}_o\bra{0}_o$, here we use the relation $\tilde{\alpha}_0=0$. The corresponding helper projector is expressed $\hat{\Pi}^{\textrm{mis}}_0 = \ket{E_0}_q\bra{E_0}_q\otimes \ket{0}_o\bra{0}_o$. The ideal projector $\hat{\Pi}_0$ is defined in Eq.~\eqref{eq:CVDVidealProjector}.

We first consider the general initial state defined in Eq.~\eqref{eq:CVDVinistate}. Using the result in Lemma \ref{lem:Fidelity_of_NOAA}, the final state of NOAA can be expressed as 
\begin{align}
    \ket{\tilde{\Psi}_f}=\tilde\gamma_l\ket{\Psi_0}-\tilde{\theta}_l\hat{\Pi}^{\text{mis}}\ket{\Psi_0},
\end{align}
and the final state dependent mismatch error defined in Eq.~\eqref{eq:lem2error} is given by
\begin{align}
    \Delta_{f} ^{2} = |\tilde\gamma_{l}|^{2}|\psi_{0}|^{2}(1-e^{-|\delta \alpha_0|^{2}})^2.
\end{align}
From the definition, the product of the projectors $\hat{\Pi}_0$ and $\hat{\Pi}_0^{\text{mis}}$ is $\hat{\Pi}^{\textrm{mis}}_{0}\hat{\Pi}_{0}\hat{\Pi}^{\textrm{mis}}_{0} = |c|^{2} \hat{\Pi}^{\textrm{mis}}_{0}$, where $c$ is the overlap between $\ket{\tilde{\alpha}_0}$ and $\ket{\alpha_0}$, we have
\begin{equation}
    |c|^{2}=|\langle \tilde{\alpha}_{0}|\alpha_{0}\rangle|^2=e^{-|\delta \alpha_0|^{2}}. \label{eq:cModulesquare}
\end{equation}
Then, Eq.~\eqref{eq:def_of_mathcalK} becomes
\begin{align}
    \cal K =& \langle \tilde{\Psi}_f|\hat{\Pi}_0^{\text{mis}}\hat{\Pi}_0\hat{\Pi}_0^{\text{mis}}|\tilde{\Psi}_f\rangle\nonumber\\
    =& |c|^2\cal F',
\end{align}
here we use the notation $\mathcal{F}'=\langle\tilde{\Psi}_f|\hat{\Pi}_0^{\text{mis}}|\tilde{\Psi}_f\rangle$. Since these two mismatched projectors satisfy $\hat{\Pi}^{\textrm{mis}}_0\hat{\Pi}^{\textrm{mis}} = \hat{\Pi}^{\textrm{mis}}\hat{\Pi}^{\textrm{mis}}_0= \hat{\Pi}^{\textrm{mis}}_0$, using the result in Lemma \ref{lem:Fidelity_of_NOAA} and Lemma \ref{lem:FPNOAA}, we have the following expression
\begin{align}
    \mathcal{F}'=&\tilde{\eta}P_L(|\psi_n|^2,|\langle\tilde{\alpha}_0|\alpha_n\rangle|^2),\label{eq:calF'}
\end{align}
where $\tilde{\eta}=\braket{\Psi_{0}|\hat{\Pi}^{\textrm{mis}}_{0}|\Psi_{0}} /\braket{\Psi_{0}|\hat{\Pi}^{\textrm{mis}}|\Psi_{0}}$. For small $\delta \alpha_0$, the relation between $\mathcal{F}'$ and the fidelity of matched NOAA  $\mathcal{F}_{\text{NOAA}}$ is
\begin{align}
    \mathcal{F}'=\mathcal{F}_{\text{NOAA}}+\mathcal{O}(|\delta\alpha_0|),
\end{align}
this is because $\langle\tilde{\alpha}_0|\alpha_n\rangle$ contains first order in $|\delta\alpha_0|$ and in general will not be canceled.

From Lemma~\ref{lem:fidelityMismatch}, fidelity is bounded by
\begin{equation}
    \begin{aligned}
     \left[\sqrt{|c|^2\mathcal{F'}} - |\tilde\gamma_{l}||\psi_{0}|(1-|c|^2)\right]^{2}\le \mathcal{F}^{\textrm{mis}}_{\text{NOAA}}\\
    \le \left[\sqrt{|c|^2\mathcal{F'}}+   |\tilde\gamma_{l}||\psi_{0}|(1-|c|^2)\right]^{2}.\label{eq:CVDVmisbound} 
    \end{aligned}
\end{equation}

In fact, for the forms of the projectors defined here, it is not hard to get an analytic expression of the mismatched NOAA's fidelity:
$\mathcal{F}_{\text{NOAA}}^{\text{mis}}=|\psi_0|^2|\tilde{\gamma}_l-|c|^2\tilde{\theta}_l|^2$, and $\mathcal{F}'=|\psi_0|^2|c|^2|\tilde{\gamma}_l-\tilde{\theta}_l|^2$. Thus, for small $\delta \alpha_0$, we have
\begin{align}
    \mathcal{F}_{\text{NOAA}}^{\text{mis}}=&|c|^2\mathcal{F}'+\mathcal{O}(|\delta \alpha_0|^2)\\
    =& \mathcal{F}_{\text{NOAA}}+\mathcal{O}(|\delta\alpha_0|)
\end{align}

For the simple two-state model defined in Eq.~\eqref{eq:2statesModel}, the expression of fidelity is simplified to
\begin{align}
    \mathcal{F}'=&\tilde{\eta}P_L\left(|\psi_0|^2,|c|^2,|\tilde{\epsilon}|,|\epsilon|\right)
\end{align}
with
$
     \tilde{\eta} = \frac{1}{1+(1-|\psi_0|^{2})|\tilde{\epsilon}|^{2}/|c|^2|\psi_0|^{2}},
$
and $\tilde{\epsilon} = \braket{\tilde{\alpha}_0|\alpha_1}$.

\begin{figure}
\centering
\includegraphics[width=\columnwidth]{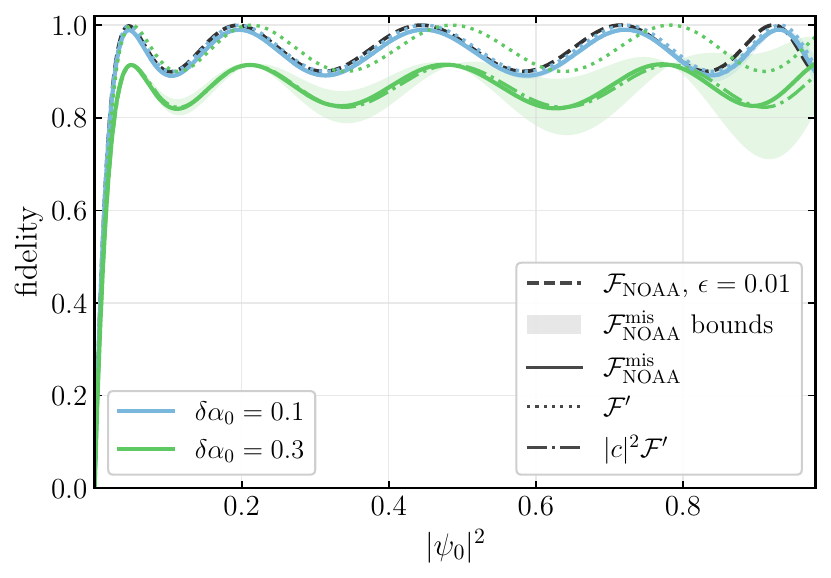}
\caption{\textbf{A comparison between matched NOAA and mismatched NOAA for two-state model.} Both NOAA sequences have degree $l=5$ and the error parameter $\delta^2=0.1$. The $\epsilon$ of matched NOAA is $0.01$. The shaded region represents the bounds defined in Eq.~\eqref{eq:CVDVmisbound}. The blue lines are mismatched NOAA with $|\delta\alpha_0|=0.1$ and the green lines are mismatched NOAA with $|\delta\alpha_0|=0.3$.}
\label{fig:2statemismatch}
\end{figure}

In Fig.~\ref{fig:2statemismatch}, we numerically calculate the mismatch effect for the two-state model. For mismatched NOAA with $|\delta \alpha_0|=0.1$, the bounds in the inequality in Eq.~\eqref{eq:CVDVmisbound} are tight, and $\mathcal{F}_{\text{NOAA}}^{\text{mis}}$ is close to $\mathcal{F}_{\text{NOAA}}$. Furthermore, $|c|^2\mathcal{F}'$ is a good approximation of  $\mathcal{F}_{\text{NOAA}}^{\text{mis}}$ for cases with small $|\delta \alpha_0|$. For  mismatched NOAA of $|\delta \alpha_0|=0.3$, the bounds are looser compared to the $|\delta\alpha_0|=0.1$, and $\mathcal{F}_{\text{NOAA}}^{\text{mis}}$ deviates from $\mathcal{F}_{\text{NOAA}}$. But $|\delta \alpha_0|^2=0.3^2=0.09$ is still small,  $|c|^2\mathcal{F}'$ can approximate $\mathcal{F}_{\text{NOAA}}^{\text{mis}}$ fairly well within a certain region of $|\psi_0|^2$. Since $\mathcal{F}'$ itself is a matched NOAA,  seen in Eq.~\eqref{eq:calF'}, we conclude that the mismatched NOAA is approximately bounded by $|c|^2\tilde{\eta}$ for small $|\delta \alpha_0|$. Thus, the mismatched NOAA also has certain fixed point feature for small deviation $|\delta \alpha_0|$. The degradation of the mismatched NOAA fidelity is mainly contributed by the overlap between $\ket{\tilde{\alpha}_0}$ and $\ket{\alpha_0}$, as presented in Eq.~\eqref{eq:cModulesquare}.

In conclusion, the non-orthogonality of the CV states imposes a fidelity ceiling on NOAA, while the mismatch effect introduces an additional reduction in fidelity. Implementing a high-fidelity reflection about a given coherent state requires exponentially increasing precision in both the knowledge of $\alpha_0$ and the implementation of the corresponding reflection gate. In the next section, we introduce an HQSP filter to deal with mismatch effect and propose an integrated solution to implement NOAA on hybrid CV-DV processors efficiently.

\section{HQSP filter for NOAA on hybrid CV-DV quantum processors}\label{sec:HQSP}

In Sec.~\ref{sec:NOAA}, we introduced the theory of NOAA and discussed the application of NOAA in the hybrid CV-DV systems, where we found that the NOAA fidelity suffers from the non-orthogonality nature of the CV states and the mismatch effect. A straightforward way to deal with non-orthogonality is to use states with smaller mutual overlaps. The ideal state would be position eigenstates $\{\ket{x}\}$ which are orthonormal, but in reality, exact position eigenstates could not be prepared. Alternatively, we use the following displaced squeezed oscillator states as basis states:
\begin{equation}
    \ket{\alpha,r}_o=D(\alpha=x/\sqrt{2})S(r)\ket{0}_o,
\end{equation}
here we use $D(\alpha)=\exp(\alpha\hat{a}^\dagger-\alpha^*\hat{a})$ and $S(r)=\exp[(r\hat{a}^2-r(\hat{a}^\dagger)^2)/2]$ with $r\ge0$, where $\hat{a}$ and $\hat{a}^\dagger$ are annihilation and creation operators of Fock states. 
In the position representation, the wavefunction of the oscillator state is
\begin{align}
    \langle x|\alpha_0,r\rangle=\frac{e^{r/2}}{\pi^{1/4}}\exp\left[-\frac{e^{2r}}{2}(x-x_0)^2\right].
\end{align}
This wavefunction has a Gaussian shape with mean position at $x_0$ and width of $\sigma=e^{-r}/\sqrt{2}$. The overlap between two states is
\begin{align}
    \langle\alpha_1,r|\alpha_2,r\rangle = \exp\left[ -\frac{e^{2r}}{2}(\alpha_1-\alpha_2)^2\right],
\end{align}
for $r>0$, the overlap will decay quicker than the unsqueezed displaced states. For the given system, by choosing the proper squeezing parameter $r$, we improve the orthogonality between CV states.
Using these basis states requires us to use a different CV state in the input state defined in Eq.~\eqref{eq:inputStateBackground}, we denote the new input state for hybrid CV-DV systems as
\begin{align}
    \ket{\Psi^r_{\text{in}}}= \sum_n \psi_n\ket{E_n}_q\otimes \ket{0,r}_o.
\end{align}
After the oracle $P_{\hat{H}}$, which is a qubit-controlled displacement introduced in Sec.~\ref{sec:background}, the output state for NOAA is 
\begin{align}
    \ket{\Psi^r_0}=\sum_n \psi_n \ket{E_n}_q\otimes\ket{\alpha_n,r}_o.\label{eq:Psi_0^r}
\end{align}
In Ref.~\cite{bellCodesigningSpectralTransformation2025}, $P_{\hat{H}}$ is a qubit-controlled displacement operation, for the more general situation, where $P_{\hat{H}}$ could be an arbitrary Gaussian operation \cite{liuHybridOscillatorqubitQuantum2026}, a real squeezing parameter $r$ may not be sufficient to generate well-separated oscillator states, one may need to use $r'=re^{i\theta}$ as the squeezing parameter. Nevertheless, the underlying idea remains unchanged, we aim to use a squeezed vacuum state
to improve the distinguishability of the CV states.

However, NOAA also suffers from the mismatch effect. For the given error $\delta\tilde{E}_t$, the quantity $|c|^2$ defined in Eq.~\eqref{eq:cModulesquare} of the new oscillator basis becomes
\begin{align}
    |c^{(r)}|^2=|\langle \tilde{\alpha}_{t},r|\alpha_{t},r\rangle|^2=\exp(-e^{2r}\alpha^2\delta\tilde{E}_t^2),
\end{align}
for $r\ge0$, we can define the effective mismatch parameter $\delta\bar{\alpha}_t=e^{r}\alpha\delta\tilde{E_t}$ becomes larger as the states become more nearly orthogonal, thus, $|c^{(r)}|$ decreases as $r$ grows, thereby leading to a greater reduction in fidelity.
In the large-squeezing regime, $\ket{\alpha_i,r}_o$ becomes sharply localized around $x_i=\sqrt{2}\alpha_i$. The error of energy estimation in the position space is expressed as
\begin{equation}
    \delta x_t = \sqrt{2}\delta \alpha_t = \sqrt{2}\alpha \delta\tilde{E}_t, \label{eq:delta_x_t}
\end{equation}
where $\delta\tilde{E}_t$ is defined in Eq.~\eqref{eq:displacementMismatch}.

\begin{figure}
    \centering
    \includegraphics[width=\columnwidth]{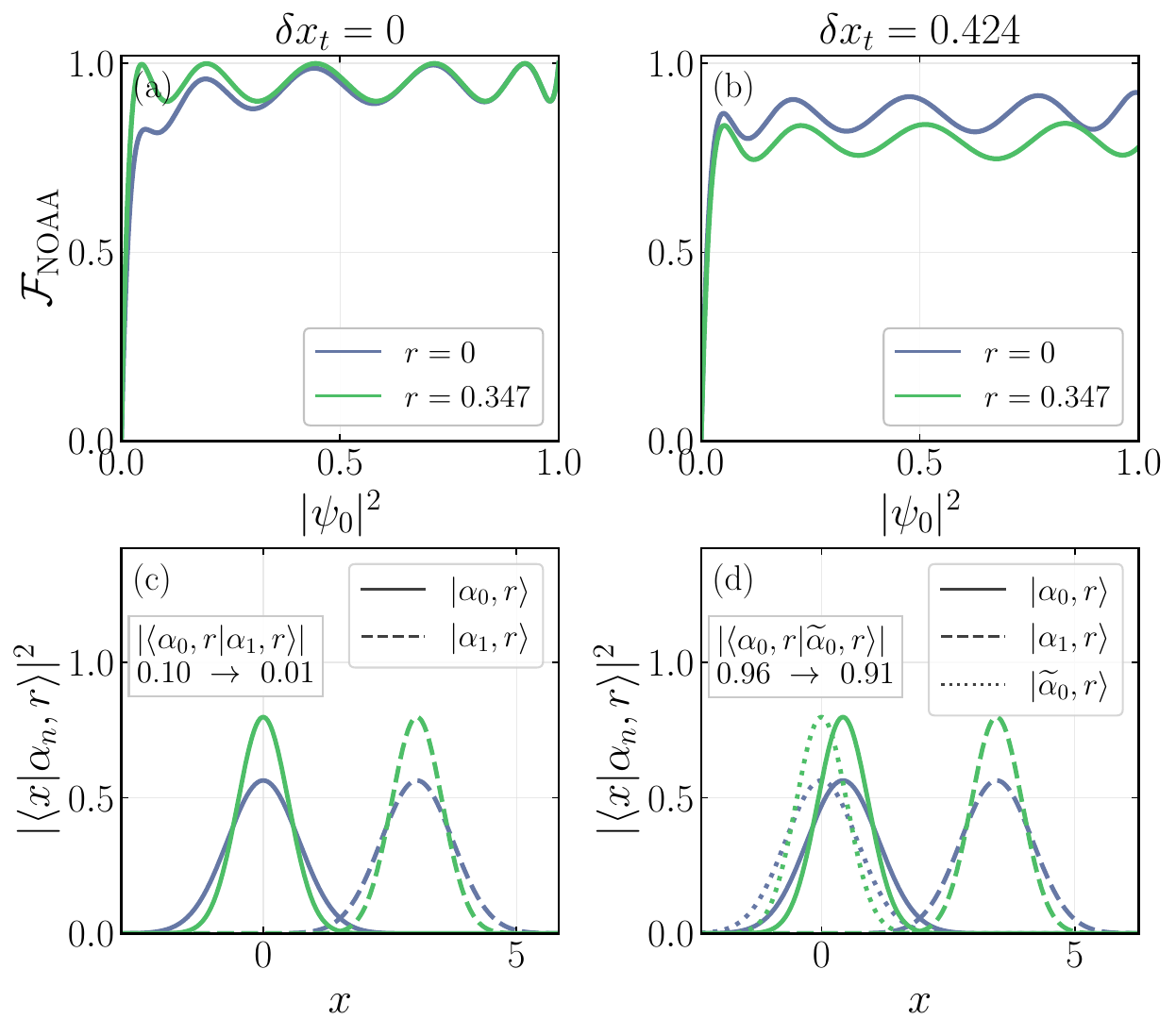}
    \caption{\textbf{Comparison between unsqueezed-state model and squeezed states model.} We consider a two-state NOAA model with sequence degree $l=5$, fixed-point error parameter $\delta^2=0.1$, we numerically calculate the fidelity with QuTiP \cite{qutip5}. The target oscillator state is centered at $x_0=0$, while the non-target state is chosen to center at $x_1\approx3.03$ such that $\epsilon=|\langle\alpha_0,r|\alpha_1,r\rangle|=0.1$ for $r=0$ with $\alpha_0=x_0/\sqrt{2}$ and $\alpha_1=x_1/\sqrt{2}$. We compare the unsqueezed case $r=0$ with the squeezed case $r=0.347$, for which $\epsilon$ decreases from $0.10$ to $0.01$. (a) and (b) compare the NOAA fidelity between unsqueezed and squeezed state for matched $\delta x_t=0$ and mismatched case $\delta x_t\neq0$. (c) and (d) are position probability density $|\langle x|\alpha_n,r\rangle|^2$, smaller overlap means the states are easier to distinguish. In the matched case, the projector $\hat{\Pi}$ is centered on the target state, corresponding to $\delta x_t=0$. In the mismatched case, both oscillator states are displaced by $\delta x_t=0.3\sqrt{2}\approx0.424$, so their centers become $x_0=\delta x_t$ and $x_1\approx3.03+\delta x_t$, while the projector remains centered at $x=0$. Consequently, the target–projector overlap magnitude decreases from approximately 0.96 in the unsqueezed case to 0.91 in the squeezed case.}
    \label{fig:squeeze_mismatch}
\end{figure}

Fig.~\ref{fig:squeeze_mismatch} shows a two-state model of the comparison between using unsqueezed states and squeezed states. We use a similar two-state model defined in Eq.~\eqref{eq:2statesModel}.  In the matched case, shown in Figs.~\ref{fig:squeeze_mismatch}(a) and \ref{fig:squeeze_mismatch}(c), increasing the squeezing from $r=0$ to $r=0.347$ narrows the oscillator wavepackets and reduces the target–non-target overlap from $0.10$ to $0.01$. The two oscillator states therefore become more distinguishable, and the NOAA fidelity approaches more closely the ideal orthogonal-state result. By contrast, in the mismatched case shown in Figs.~\ref{fig:squeeze_mismatch}(b) and \ref{fig:squeeze_mismatch}(d), the projector center is displaced from the true target center by the same absolute position error for both squeezing parameters. Because the squeezed target state has a smaller position-space width, the same displacement corresponds to a larger error relative to the wavepacket width, reducing the target–projector overlap from approximately $0.96$ to $0.91$.

Therefore, squeezing enhances not only the orthogonality between different oscillator states, but also the sensitivity to a fixed energy-estimation error. Increasing $r$ alone would lead to a tradeoff between suppressing overlaps of different basis states and maintaining robustness against projector mismatch. Consequently, high-fidelity NOAA cannot, in general, be achieved by squeezing alone. To overcome this limitation, we introduce the HQSP filter, which provides an additional energy-selective suppression of the undesired components and a tolerance window for projector mismatch. Importantly, within the HQSP construction, squeezed displaced states are not merely used to improve the orthogonality of the oscillator encoding, they are an essential resource for implementing the filter itself.

The remainder of this section is organized as follows. In Sec.~\ref{subsec:SetupHQSP}, we introduce the basic setup of HQSP and provide the construction of the HQSP filter for NOAA. In Sec.~\ref{subsec:NumericalResults}, a numerical simulation is used for comparison between NOAA in the CV-DV systems and RUS used in Ref.~\cite{bellCodesigningSpectralTransformation2025}.

\subsection{HQSP filter}\label{subsec:SetupHQSP}
In this section, we introduce an HQSP filter to suppress the non-target components in the initial state, thereby improving the NOAA fidelity.

\begin{figure}
\centering
\includegraphics[width=\columnwidth]{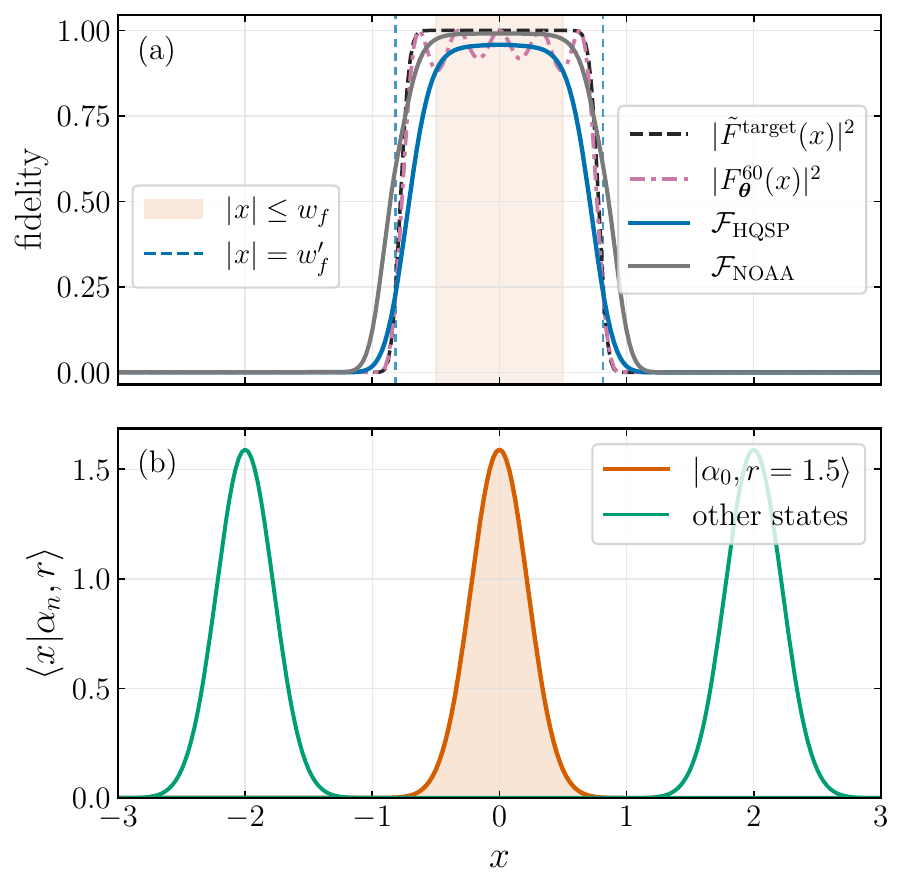}
\caption{\textbf{HQSP filter and NOAA.} (a) The numerical result of the target function $\tilde{F}^{\text{target}}$, the encoded polynomial $F_{\boldsymbol{\theta}}^{60}$, the fidelity of HQSP $\mathcal{F}_{\text{HQSP}}$ and NOAA $\mathcal{F}_{\text{NOAA}}$, with QuTiP \cite{qutip5}. The target function $\tilde{F}^{\text{target}}(x)$ is defined in Eq.~\eqref{eq:HQSP_target_func}. $F_{\boldsymbol{\theta}}^{60}(x)$ is the degree-60 polynomial encoded in HQSP, defined in Eq.~\eqref{eq:F_boldtheta}. We use an input state $\ket{\chi}=\ket{0}_q\otimes\ket{\alpha,r}_o \otimes\ket{0}_Q$ with $\alpha= x/\sqrt{2}$ and $r=1.5$ to calculate the fidelity. The fidelity of applying HQSP alone and HQSP with NOAA are therefore defined as $\mathcal{F}_{\text{HQSP}}=|\langle \chi|\tilde{U}_{\text{HQSP}}|\chi \rangle|^2$ and $\mathcal{F}_{\text{NOAA}}=|\langle \chi|\hat{U}_{L}\tilde{U}_{\text{HQSP}}|\chi \rangle|^2$. The details of circuits $\hat{U}_{L}$ and $\tilde{U}_{\text{HQSP}}$ can be found in the discussion in Sec.~\ref{sec:NOAA} and Sec.~\ref{sec:HQSP}. The parameters of HQSP are set as $ d = 60$, $\epsilon_\text{sm}=0.33$, $T=10$ and target width is $w_f=0.5$, and the parameters of NOAA are $\delta^{2} = 0.01, \, l=2$. (b) Illustration of Gaussian states. The orange curve is the target state centered at $x=0$, green curves are states we want to filter out. From (a) we can see that the quasi-fixed-point behavior of NOAA allows a moderate-degree HQSP filter to achieve high final fidelity.}
\label{fig:HQSP_filter}
\end{figure} 

In Sec.~\ref{ssec:hqsp}, we introduced the basic setting of HQSP. To construct the filter, we choose $F(\hat{w})$ to be a real-valued Laurent polynomial and $F(\hat{w})=F(\hat{w}^{-1})$. With $\hat{w}=\exp(-i\kappa\hat{x}/2)$, we denote the corresponding response in the position variable by $\tilde{F}(\hat{x})=F(\hat{w})$. We denote the HQSP circuit with this setting as $
\tilde{U}_{\text{HQSP}}$. Because $\hat{w}$ is a periodic function of $\hat{x}$, $\tilde{F}(\hat{x})$ is a periodic function of $\hat{x}$. We therefore need to carefully set the period $T$ of the function and choose the rescaling parameter $\kappa$ such that no undesired oscillator component lies inside these replicated passbands. 

The desired polynomial transformation is block encoded in the upper left block of $\tilde{U}_{\text{HQSP}}$. Here we introduce an auxiliary qubit $\ket{0}_Q$ to implement such transformation. 
The initial state for HQSP sequence is 
\begin{align}
    \ket{\Psi_\text{in}^{\text{HQSP}}}=&\hat{U}_{\text{prep}}^{H}\ket{0}_q\otimes\ket{0}_o\otimes\ket{0}_Q=\ket{\Psi^r_0}\otimes\ket{0}_Q\label{eq:HQSPinitial_state}
\end{align}
where $\hat{U}_{\text{prep}}^{H}$ is the preparation circuit, $\ket{\Psi_0^r}$ is defined in Eq.~\eqref{eq:Psi_0^r}.
The output state of the HQSP sequence is
\begin{equation}
\begin{aligned}
|\Psi_H\rangle\approx{}&\sum_i \tilde F(x_i)\psi_i|E_i\rangle_q|\alpha_i,r\rangle_o|0\rangle_Q\\
&-\sum_i e^{i\gamma_i}\sqrt{1-\tilde F(x_i)^2}\,\psi_i|E_i\rangle_q|\alpha_i,r\rangle_o|1\rangle_Q ,
\end{aligned}
\end{equation}
where $\gamma_i$ denotes the phase of the complementary HQSP amplitude near $x_i=\sqrt{2}\alpha_i$. The approximation is valid when the HQSP response varies negligibly over the position-space width $\sigma=e^{-r}/\sqrt{2}$ of $\ket{\alpha_i,r}_o$. 

The goal of the HQSP filter is to filter states that center around $x=0$ and the filter bandwidth should be sufficiently broad to account for the error $\delta x_t$, i.e.
\begin{align}
    \tilde{F}(x)=\begin{cases}
        1, x\in(-w_f,w_f)\\
        0, |x|\in (w_f,T/2]
    \end{cases}\label{eq:exact_filter}
\end{align}
where $w_f=\sqrt{2}\delta\alpha_t^{\text{max}}$ is the half-width of the passband, $\delta\alpha_t^{\text{max}}$ as the upper bound of $|\delta\alpha_t|$ and $T$ is the period of the filter.

The ideal filter in Eq.~\eqref{eq:exact_filter} can be written, over one fundamental interval, as a linear combination of two shifted sign functions,
\begin{equation}
    \tilde{F}^{\text{ideal}}(x)=\frac{1}{2}\left[\text{sign}(x+w_f)-\text{sign}(x-w_f)\right],
\end{equation}
where $x\in[-T/2,T/2]$. Since the sign function is discontinuous, it cannot be represented exactly by a finite-degree HQSP sequence. We therefore use the error function as a smooth approximation~\cite{martynEfficientFullycoherentQuantum2023},
\begin{equation}
    \text{erf}(kx)=\frac{2}{\sqrt{\pi}}\int_{0}^{kx} e^{-y^2}\mathrm{d}y=\frac{2k}{\sqrt{\pi}}\int_{0}^{x} e^{-(ku)^2}\mathrm{d}u ,
\end{equation}
where $k>0$ controls the steepness of the transition.

Because the displaced squeezed states have a finite position-space width $\sigma=e^{-r}/\sqrt{2}$, in the numerical implementation we enlarge the nominal passband half-width to $w_f' = w_f + 2\sigma$, and choose the transition width as $\Delta_H=\sigma$. The smooth target response used for the HQSP phase optimization is therefore
\begin{equation}
    \tilde{F}^{\mathrm{target}}(x)=\frac{1}{2}\left[\text{erf}\!\left(k(x+w_f')\right)-\text{erf}\!\left(k(x-w_f')\right)\right].\label{eq:HQSP_target_func}
\end{equation}
We choose $k$ such that the error-function approximation reaches the specified tolerance $\epsilon_\text{sm}$ outside the transition region. Following the construction in Ref.~\cite{martynEfficientFullycoherentQuantum2023}, we set
$k=\sqrt{2W\left(2/\pi\epsilon_\text{sm}^2\right)}/\Delta_H$, where $W$ denotes the principal branch of the Lambert $W$ function. Here, $\epsilon_\text{sm}$ specifies the smoothing tolerance used to construct the target response.

We directly optimize the HQSP phase factors against samples of $\tilde{F}^{\mathrm{target}}(x)$. For the phase optimization, we use the symmetric phase factors construction of Ref.~\cite{dongEfficientPhasefactorEvaluation2021}. For a degree-$d$ HQSP sequence, we define $\tilde{d}=\left\lceil\frac{d+1}{2}\right\rceil$ independent phase parameters $\boldsymbol{\phi}=(\phi_0,\ldots,\phi_{\tilde{d}-1})$. These parameters are reflected to form symmetric phase factors $\boldsymbol{\varphi}=(\varphi_0,\ldots,\varphi_d)$ satisfying $\varphi_j=\varphi_{d-j}$ and $\varphi_i=\phi_i$ for $i=0,\ldots,\tilde{d}-1$. The phase factors used in the HQSP sequence are then defined as $\boldsymbol{\theta}=\boldsymbol{\varphi}+(\frac{\pi}{4},0,0,\ldots,0,-\frac{\pi}{4})$, which shifts the first and the last angles of $\boldsymbol{\varphi}$ by $+\pi/4$ and $-\pi/4$, respectively. This is equivalent to $ \tilde{U}_{\text{HQSP}}(\hat{x};\boldsymbol{\theta}) = e^{i\pi X_Q/4} \tilde{U}_{\text{HQSP}}(\hat{x};\boldsymbol{\varphi}) e^{-i\pi X_Q/4}$. The HQSP sequence with the phase factors $\boldsymbol{\varphi}$ has the Pauli decomposition $\tilde U_{\mathrm{HQSP}}(\hat{x};\boldsymbol{\varphi}) =\int \mathrm{d}x\,\ket{x}_o\bra{x}_o\otimes\left[A(x)I_Q+iB(x)X_Q+iC(x)Z_Q\right]$ with $A(x),\,B(x),\,C(x)\in\mathbb{R}$ and $X_Q$, $Y_Q$, $Z_Q$ are the Pauli matrices acting on the ancilla qubit space \cite{dongEfficientPhasefactorEvaluation2021}. The corresponding HQSP matrix elements defined in Eq.~\eqref{eq:HQSPBackground} are $F(x)=A(x)+iC(x)$ and $iG(x)=iB(x)$. Since $e^{i\pi X_Q/4}Z_Qe^{-i\pi X_Q/4}=Y_Q$, the endpoint shift gives
$\widetilde U_{\mathrm{HQSP}}(\hat{x};\boldsymbol{\theta})
=\int \mathrm{d}x\,\ket{x}_o\bra{x}_o\otimes\left[A(x)I_Q+iB(x)X_Q+iC(x)Y_Q\right]$, such that $F(x)=A(x)$ and $iG(x)=C(x)+iB(x)$ ensuring that the upper-left HQSP response is real valued.

To construct the fitting points, we introduce the real QSP variable $z=\frac{w+w^{-1}}{2}=\cos\left(\frac{\kappa x}{2}\right)$ and use the Chebyshev nodes \cite{dongEfficientPhasefactorEvaluation2021} $z_j= \cos\left(\frac{(2j-1)\pi}{4\tilde{d}}\right),\,\,j=1,\ldots,\tilde{d}$ with the corresponding position coordinates $x_j=\frac{2}{\kappa}\arccos z_j$.

For a given phase sequence, we denote the upper-left HQSP amplitude by
\begin{equation}
    F_{\boldsymbol{\theta}(\boldsymbol{\phi})}(x)=\bra{x}_o\bra{0}_Q\tilde{U}_{\mathrm{HQSP}}\left(\hat{x};\boldsymbol{\theta}\right)\ket{x}_o\ket{0}_Q,\label{eq:Fboldphi}
\end{equation}
where $\ket{x}_o$ is the position eigenstate.
The phase angles are obtained by minimizing
\begin{equation}
    \mathcal{L}\left( \boldsymbol{\phi} \right)=\frac{1}{2\tilde d}\sum_{j=1}^{\tilde d}\left|F_{\boldsymbol{\theta}(\boldsymbol{\phi})}(x_j)-\tilde{F}^{\mathrm{target}}(x_j)\right|^2 . \label{eq:optL}
\end{equation}
We minimize $\mathcal{L}(\boldsymbol{\phi})$ using the L-BFGS-B algorithm implemented in SciPy~\cite{2020SciPy-NMeth}.

The optimized phase sequence obtained from Eq.~\eqref{eq:optL} determines a degree-$d$ HQSP response $F_{\boldsymbol{\theta}}^d(x)$. For the real-valued filter considered here, the filter response satisfies $F_{\boldsymbol{\theta}}^d(w)=F_{\boldsymbol{\theta}}^d(w^{-1})$. Using the QSP variable $z$, the identity $T_n(z)=(w^n+w^{-n})/2$ shows that the corresponding Laurent polynomial can be written as a Chebyshev expansion,
\begin{equation}
    F_{\boldsymbol{\theta}}^d(x)=\sum_{n=0}^{d} c_n T_n(z),\label{eq:F_boldtheta}
\end{equation}
where $T_n(z)$ denotes the Chebyshev polynomial of the first kind. We also verify this polynomial representation numerically for the HQSP phases obtained from the optimization. The QSP parity constraint implies $c_n=0$ for $n \neq d(\text{mod}2)$, so only Chebyshev orders with the same parity as \(d\) appear.

The degree required to approximate the sign-function transition can therefore be estimated using the standard QSP polynomial-approximation result. For a normalized transition width $\Delta$ and approximation error $\epsilon_\text{poly}$, a sign-function approximation can be constructed with degree $d=\mathcal{O}\left[\frac{1}{\Delta}\log\left(\frac{1}{\epsilon_\text{poly}}\right)\right]$ \cite{martynEfficientFullycoherentQuantum2023}. Here, $\epsilon_{\text{poly}}$ denotes the polynomial-approximation error and is distinct from the smoothing tolerance $\epsilon_{\text{sm}}$ used to construct the target response.
    
For the HQSP filter, the normalized transition width is set by the physical transition width $\Delta_H$ relative to the filter period $T=\Theta(1/\kappa)$, giving $\Delta=\Theta(\kappa\Delta_H)=\Theta(\Delta_H/T)$, and hence
\begin{equation}
    d=\mathcal{O}\left[\frac{T}{\Delta_H}\log\left(\frac{1}{\epsilon_\text{poly}}\right)\right]=\mathcal{O}\left[\frac{1}{\kappa\Delta_H}\log\left(\frac{1}{\epsilon_\text{poly}}\right)\right].
\end{equation}

In Fig.~\ref{fig:HQSP_filter}, we present the numerical simulation of the filter with $T=10$, $\epsilon_\text{sm}=0.33$, $r=1.5$, $w_f=0.5$, and $d=60$. The corresponding HQSP phase factors are obtained directly from the optimization defined above. Because the subsequent NOAA sequence exhibits quasi-fixed-point behavior, a moderate-degree HQSP filter is sufficient to obtain high final fidelity. We therefore use $d=60$ throughout the numerical simulations in this section.

Now, let us examine the ideal fidelity of NOAA with HQSP. We assume $\ket{E_0}_q\ket{\alpha_0,r}_o\ket{0}_Q$ is the target state. For simplicity, we characterize the ideal HQSP filtering error by a uniform leakage amplitude $\vartheta$. Then the filtered state can be written as:
\begin{equation}
\begin{aligned}
|\Psi_H\rangle=\,&\psi_0 |E_0\rangle_q |\alpha_0,r\rangle_o\left(\sqrt{1-\vartheta^2}\,|0\rangle_Q-e^{i\gamma_0}\vartheta\,|1\rangle_Q\right)\\
&+\sum_{n\neq 0}\psi_n |E_n\rangle_q |\alpha_n,r\rangle_o\\
&\times\left(\vartheta\,|0\rangle_Q-e^{i\gamma_n}\sqrt{1-\vartheta^2}\,|1\rangle_Q\right),
\end{aligned}
\end{equation}
Here the phases $\gamma_0$ and $\gamma_n$ are irrelevant for the NOAA projector, since the success probability depends only on the $|0\rangle_Q$ branch.

We define projector of NOAA as $\hat{\Pi}=\hat{I}_q\otimes\hat{I}_o\otimes\ket{0}_Q\bra{0}_Q$, and the ideal projector $\hat{\Pi}_0=\ket{E_0}_q\bra{E_0}_q\otimes\ket{\alpha_0,r}_o\bra{\alpha_0,r}_o\otimes\ket{0}_Q\bra{0}_Q$, such that $\braket{ \Psi_{H}| \hat{\Pi}| \Psi_{H} }=(1-\vartheta^2)|\psi_{0}|^{2}+\vartheta^2(1-|\psi_{0}|^{2})$, and $\braket{ \Psi_{H}| \hat{\Pi}_{0}| \Psi_{H} }=(1-\vartheta^2 ) |\psi_{0}|^{2}$.
The factor $\eta$ can be expressed as
\begin{align}
  \eta=&\frac{(1-\vartheta^2)|\psi_{0}|^{2}}{(1-\vartheta^2)|\psi_{0}|^{2}+\vartheta^2(1-|\psi_{0}|^{2})}  ,\label{eq:hqsp_eta}
\end{align}
which approaches $1$ as $\vartheta\to 0$.
Since the HQSP defined a filter with a window around the target $\alpha_0$, in mismatched case, the upper bound of the fidelity becomes Eq.~\eqref{eq:hqsp_eta}. This demonstrates that HQSP filter can help to solve the reduction of fidelity brought by the mismatch effect.

\subsection{Numerical Results}\label{subsec:NumericalResults}

\begin{figure*}[!htbp]
\centering
\includegraphics[width=2\columnwidth]{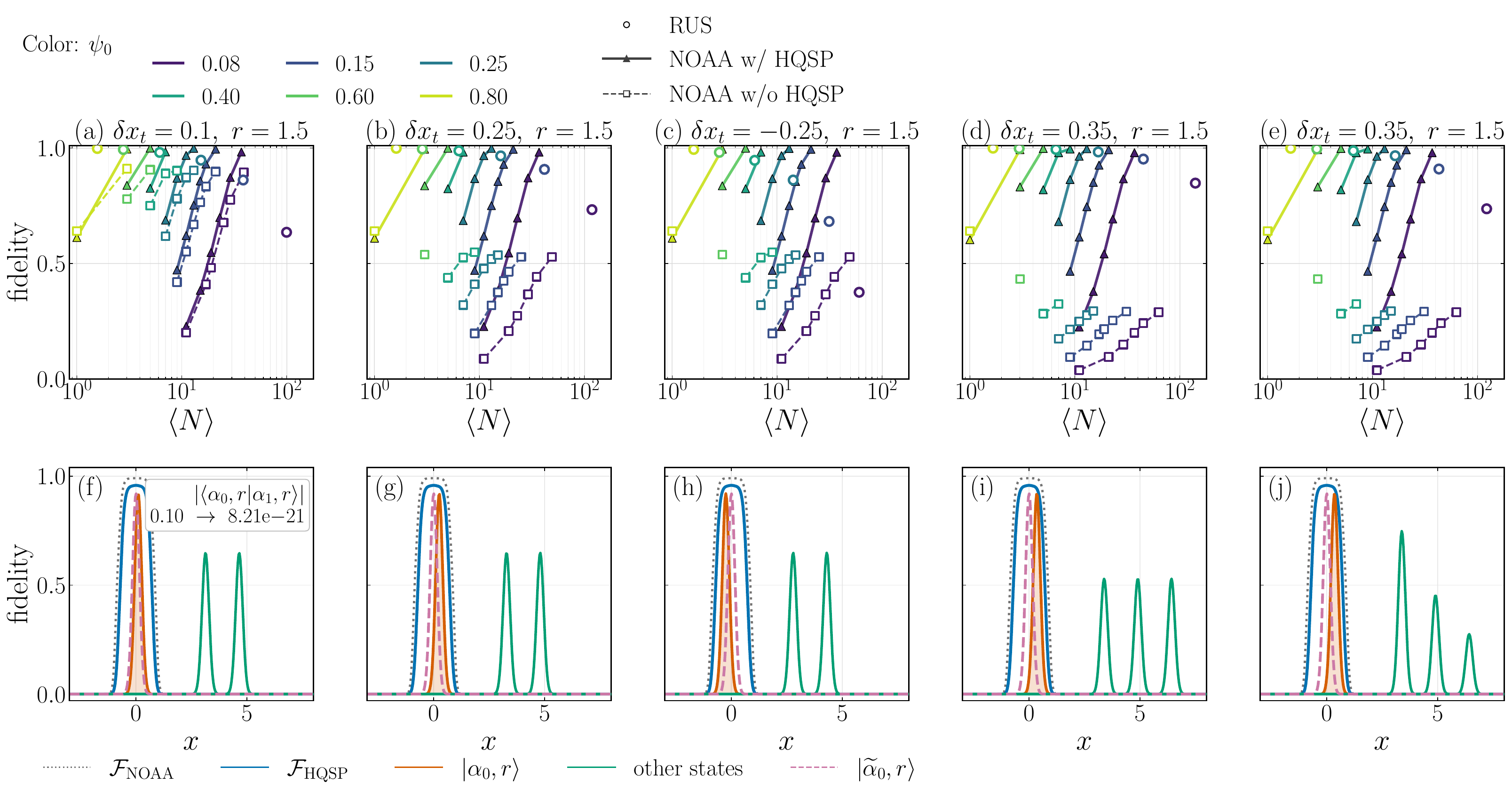}
\caption{\textbf{Comparison of RUS and NOAA performance for different target-energy mismatches and initial-state distributions.} Panels (a)–(e) show the target-state fidelity as a function of the average number of state-preparation queries $\langle N\rangle$. Colors denote different initial target amplitudes $\psi_0$. Open circles represent RUS, filled triangles represent NOAA preceded by the HQSP filter, and open squares represent NOAA without HQSP. The parameters of HQSP are $d = 60$, $\epsilon_\text{sm}=0.33$, $T=10$ and target width is $w_f=0.5$. We use the phase angles of NOAA given in Eq.~\eqref{eq:YoderPhase} with $\delta^2=0.01$. We choose the squeezing parameter $r=1.5$ for initial state $\ket{\Psi_0^r}$. Panels (f)–(j) show the corresponding oscillator wavepackets and filter responses. The solid blue curve shows the HQSP fidelity for an input squeezed wavepacket centered at \(x\), while the dotted gray curve shows the corresponding fidelity after HQSP followed by NOAA. The orange and dashed magenta curves represent the target oscillator state and the mismatched projector state, respectively, and the green curves represent non-target components. The height of the other states shows the ratio of their coefficient in the initial state. As the target mismatch increases, bare NOAA develops a strong fidelity ceiling, whereas HQSP-assisted NOAA retains high fidelity provided that the target wavepacket remains inside the HQSP passband. In panels (a-c) and corresponding (f-h), the center of oscillator states are set as $x_0=\delta x_t$, where $\delta x_t$ is defined in Eq.~\eqref{eq:delta_x_t}, $x_1\approx3.034+\delta x_t$ and $x_2\approx4.552+\delta x_t$. Panels (d) and (e) further illustrate the dependence on the number and amplitude distribution of non-target states. The center of oscillator states in (d-e) and (i-j) are $x_0=\delta x_t$, $x_1\approx3.034+\delta x_t$, $x_2\approx4.552+\delta x_t$ and $x_3\approx6.070+\delta x_t$.}
\label{fig:NOAAvsRUS}
\end{figure*}

In this section, we numerically compare the query complexity of RUS with the filtered NOAA.
In Sec.~\ref{ssec:filter}, we showed that the success probability $P_{s}$ of RUS shown in Eq.~\eqref{eq:RUS_prob} and the corresponding post-measurement state is in Eq.~\eqref{eq:RUS_state}. We clarify that, for RUS, we always use unsqueezed states as the input.

The average query complexity $\langle N_{\text{RUS}}\rangle$ of successfully postselecting the vacuum state is
\begin{align}
    \langle N_{\text{RUS}}\rangle&=\sum_{n=1}^{\infty} n (1-P_{s})^{n-1}P_{s}=\frac{1}{P_{s}}.
    \label{eq:N_RUS}
\end{align}
The fidelity of RUS is 
\begin{align}
    \mathcal{F}_{\text{RUS}}=|\braket{E_0|\Psi_{\text{RUS}}}|^2 
    =|\psi_0|^2 \,e^{-|\alpha_0|^2} \langle N_{\text{RUS}}\rangle. \label{eq:FidelityRUS}
\end{align}

In the numerical simulation, we use $\mathcal{F}_{\text{RUS}}$ to evaluate the performance of RUS. Since the average cost $ \langle N_{\text{RUS}}\rangle$ is determined by the initial state, the fidelity is fixed for a given initial state and varies as the initial overlap $\psi_0$ changes.

To compare the resource costs, let $N$ denote the number of sites in the target DV system and $d$ the degree of the HQSP filter. For the local spin-model compilation considered in Ref.~\cite{bellCodesigningSpectralTransformation2025}, the circuit depth of the CV-DV spectral-transformation oracle scales linearly with the spatial extent of the target system. If $m$ product-formula steps are required for the desired implementation accuracy, we therefore write its cost as $C_{P_{\hat H}}=\mathcal{O}(mN)$ for a one-dimensional local system. In contrast, a degree-$d$ single-variable HQSP sequence contains $d$ conditional signal operations and $d+1$ single-qubit rotations \cite{liuMixedAnalogdigitalQuantum2025a}, giving $C_{\mathrm{HQSP}}=\mathcal{O}(d)$, while the squeezing operation contributes $\mathcal{O}(1)$ operations with respect to $N$.

Therefore, for fixed HQSP degree and squeezing strength, the preparation oracle asymptotically dominates the circuit cost when $mN\gg d$. Here, $d$ is treated as fixed for the spectral window considered. More generally, the required HQSP degree may depend on the encoded spectral bandwidth through the choice of the filter period $T$. In this regime, we use the number of applications of the preparation oracle as the query complexity and neglect the additional HQSP and squeezing overhead.

For RUS, the expected cost is given in Eq.~\eqref{eq:N_RUS}. For NOAA, following the query convention of amplitude amplification, we count $\langle N_{\mathrm{NOAA}} \rangle=2l+1$, where $l$ is the number of generalized Grover iterates.

The initial state for NOAA used in numerical simulation has a form similar to Eq.~\eqref{eq:HQSPinitial_state} and we set $r=1.5$ as the squeezing parameter. The parameters of HQSP are $d = 60$, $\epsilon_\text{sm}=0.33$, $T=10$ and target width is $w_f=0.5$. For the simulations in Fig.~\ref{fig:NOAAvsRUS}, the oscillator Hilbert space is truncated at $N_{\mathrm{BOS}}=400$ with QuTiP \cite{qutip5}. We verified convergence against $N_{\mathrm{BOS}}=600$, with the maximum absolute deviation of the plotted fidelities below $6\times10^{-6}$. The numerical results are shown in Fig.~\ref{fig:NOAAvsRUS}. As shown in the previous section, squeezing would reduce the overlap between target state and non-target states. For the parameters considered here, the overlap between CV states $\ket{\alpha_0,r}_o$ and the first excited-state component $\ket{\alpha_1,r}_o$ decreases from $0.1$ to $8.21\times10^{-21}$ as $r$ increases from $0$ to $1.5$, making the two states nearly orthogonal. For matched case, $\delta x_t=0$, the NOAA without HQSP filter and the NOAA with HQSP filter reach similar performance, and share the same fidelity for each $\psi_0$ and query cost $\langle N\rangle$. However, for slightly mismatched situation, in which $\delta x_t=0.1$, the NOAA fidelity without HQSP filter, denoted as $\mathcal{F}_{\text{NOAA}}^{noH}$, will suffer the mismatch effect, and the NOAA fidelity with HQSP filter, denoted as $\mathcal{F}_{\text{NOAA}}^H$, is robust against such small deviation, as shown in Fig.~\ref{fig:NOAAvsRUS}(a).  For RUS, in the model used for simulation, we can assume the overlaps between states of higher $n$, e.g., $\ket{\alpha_2,r}_o$ in Fig.~\ref{fig:NOAAvsRUS}(a-e) and $\ket{\alpha_3,r}_o$ in Fig.~\ref{fig:NOAAvsRUS}(d,e), and the target state can be considered as $0$ even for $r=0$. We define
\begin{align}
    \mathcal{E}\approx \frac{|\psi_1|^2}{|\psi_0|^2}e^{-\left(\alpha_1^2-\alpha_0^2\right)},
\end{align}
and $\mathcal{F}_{\text{RUS}}$ can be expressed as
\begin{align}
    \mathcal{F}_{\text{RUS}}\approx\frac{1}{1+\mathcal{E}}.\label{eq:FRUSwithE}
\end{align}
This shows that, for given $x_0$, $x_1$, $\mathcal{E}$ decreases as $|\psi_0|$ increases, and fidelity $\mathcal{F}_{\text{RUS}}$ becomes higher. Fig.~\ref{fig:NOAAvsRUS}(a,b,d) shows such relation. For $\psi_0=0.8,$ $0.6$, where the occupation of the target state in the initial state is large, RUS can reach a high fidelity with few preparation-oracle queries and performs better than both NOAA with filter and without filter. In such situation, RUS is adequate. But for small $\psi_0$, e.g., $\psi_0=0.25$, $0.15$, $0.08$, from $\mathcal{E}$, $\mathcal{F}_{\text{RUS}}$ will be lower, and the success probability will decrease, leading to a higher cost. However, NOAA can still reach a high fidelity within smaller $\langle N \rangle$ for small $\psi_0$. Besides this, for some $\psi_0$, the fidelity of RUS may be low, NOAA can help to reach higher fidelity at a similar cost or slightly more query times. 

The fidelity of RUS is also sensitive to the target-energy mismatch. Within the approximation of Eq.~\eqref{eq:FRUSwithE}, where the dominant non-target contribution comes from the first excited state, a positive mismatch $\delta x_t>0$ increases $\alpha_1^2-\alpha_0^2$ and therefore reduces $\mathcal{E}$, leading to a higher $\mathcal{F}_{\mathrm{RUS}}$. Conversely, a negative mismatch $\delta x_t<0$ decreases $\alpha_1^2-\alpha_0^2$, thereby increasing $\mathcal{E}$ and reducing $\mathcal{F}_{\mathrm{RUS}}$. This asymmetric dependence on the sign of the mismatch can be seen by comparing Figs.~\ref{fig:NOAAvsRUS}(a)--(c).
In contrast, HQSP-assisted NOAA is largely insensitive to the value and sign of \(\delta x_t\) provided that the target wavepacket remains within the filter passband.

For the parameter regime considered here, the higher excited-state components have negligible overlap with the target state, so their main effect is to redistribute the initial weight among the non-target components. An example of a model with three non-target states is presented in Fig.~\ref{fig:NOAAvsRUS}(d-e) and (i-j). 
The amplitude of $\ket{\alpha_1,r}_o$ is smaller in (d-e) compared to (a-c), so $|\psi_1|^2/|\psi_0|^2$ decreases and correspondingly $\mathcal{F}_{\text{RUS}}$ becomes large. In (e), the amplitude of $\ket{\alpha_1,r}_o$ is higher than the amplitude in (d), from $\mathcal{E}$, the fidelity of RUS is smaller. 
It seems that mismatch effect will help RUS to reach higher fidelity, however, according to Eq.~\eqref{eq:FidelityRUS}, for fixed $|\psi_0|$ and $x_0$, higher fidelity will lead to higher cost $\langle N \rangle$, this is checked by the numerical results in Fig.~\ref{fig:NOAAvsRUS}(d) and (e). This actually supports the use of NOAA with HQSP filter. The advantage of NOAA with HQSP is more obvious in (d) and (e) for small $\psi_0$. 

In summary, the HQSP filter protects NOAA against target-energy mismatch within its passband. For small initial target amplitudes, HQSP-assisted NOAA reaches high fidelity with fewer preparation-oracle queries than RUS. For large target amplitudes, RUS is already efficient and amplitude amplification provides little additional advantage.

\section{Conclusion}\label{sec:conclusion}
In this work, we introduce the non-orthogonal amplitude amplification (NOAA) algorithm for hybrid CV-DV systems, which is designed to amplify the amplitude of a target state in the presence of non-orthogonal states. We derive the NOAA fidelity and show that it is bounded by a ratio defined by the initial state and the projector used in NOAA. Moreover, we show that NOAA also suffers from the mismatch effect, which will restrict the performance of NOAA in realistic applications. We therefore introduce an HQSP filter to mitigate the mismatch effect and improve the performance of NOAA. Finally, we compare the performance of NOAA with the repeat-until-success method (RUS) and show that NOAA can achieve higher fidelity with lower query complexity in certain parameter regimes. Our work provides a framework for designing amplitude amplification algorithms in hybrid CV-DV systems and highlights the importance of considering non-orthogonality in such systems.

Several directions remain open. Since amplitude amplification can be formulated as a particular quantum signal processing sequence, a natural extension of the present work is to develop a more general QSP framework in which the signal reflection is implemented through a nonselective or imperfectly matched reflection operator. Such a formulation may characterize the class of polynomial transformations achievable with restricted information about the target state and clarify which errors can be suppressed by increasing the polynomial degree and which are fundamentally limited by the accessible projector. The intrinsic periodicity of HQSP suggests extending NOAA to periodic or multi-window spectral problems, where replicated passbands may be exploited rather than merely avoided. More generally, sequences constructed from multiple HQSP-engineered approximate projectors may provide a route toward robust amplitude amplification using only coarse spectral information, and may potentially overcome the fidelity ceiling associated with a fixed projector.

\begin{acknowledgments}
We thank Eugene Dumitrescu, Steven Girvin, Yan Wang, and Nathan Wiebe for helpful discussions. This work is supported by the U.S. Department of Energy, Office of Science, Advanced Scientific Computing Research, under contract number DE-SC0025384.
\end{acknowledgments}

\section*{Data Availability}
The code and data that support the findings of this article are
openly available at \cite{noaa_cv_dv_2026}; embargo periods may apply.

\appendix
\section{Proof of Lemma 1}\label{apendix:Lemma1}
Let us first prove that the final state after NOAA lies in the subspace spanned by {$\ket{\Psi_0},\hat{\Pi}\ket{\Psi_0}$}. For $l=0$, we have $\hat{U}_{L=1}=\hat{I}$. Assume for $l=n$, there exist $\gamma_{n}, \theta_{n} \in \mathbb{C}$, such that 
\begin{equation}
    \hat{U}_{L=2n+1}\ket{\Psi_{0}}=\gamma_{n}\ket{\Psi_{0}}-\theta_{n}\hat{\Pi}\ket{\Psi_0}.
\end{equation}
Then for $l=n+1$, denote $A_{l}=1-e^{-i\alpha_{l}}, \, B_{l} = 1 - e^{i \beta_{l}}$, we have
\begin{align}
    \hat{U}_{2n+3}\ket{\Psi_0} =& \hat{G}(\alpha_{n+1},\beta_{n+1}) \hat{U}_{2n+1}\ket{\Psi_0} \nonumber\\
    =& -S_s(\alpha_{n+1})S_{t}(\beta_{n+1})(\gamma_{n}\ket{\Psi_{0}}-\theta_{n}\hat{\Pi}\ket{\Psi_0}),
\end{align}
\begin{align}
    S_t(\beta_{n+1})\hat{U}_{2n+1}\ket{\Psi_{0}}=& (\hat{I}-B_{n+1}\hat{\Pi})(\gamma_{n}\ket{\Psi_{0}}-\theta_{n}\hat{\Pi}\ket{\Psi_0})\nonumber\\
    =&\bigg[\gamma_{n}\ket{\Psi_0}\nonumber\\
    &-(\gamma_{n}B_{n+1}+\theta_{n}-\theta_{n}B_{n+1})\hat{\Pi}\ket{\Psi_0}\bigg].
\end{align}
The state after applying $\hat{U}_{2n+3}$ is
\begin{align}
    \hat{U}_{2n+3}\ket{\Psi_0}=& -(I-A_{n+1}\ket{\Psi_0}\bra{\Psi_0})\bigg[\gamma_{n}\ket{\Psi_0}\nonumber\\
    &-(\gamma_{n}B_{n+1}+\theta_{n}e^{i\beta_{n+1}})\hat{\Pi}\ket{\Psi_0}\bigg] \nonumber\\
    =& -\big(e^{-i\alpha_{n+1}}\gamma_{n}+A_{n+1}B_{n+1}\gamma_{n}\zeta \nonumber \\
    &+e^{i\beta_{n+1}}A_{n+1}\theta_{n}\zeta\big)\ket{\Psi_{0}} \nonumber\\
    &+\big(\gamma_{n}B_{n+1}+e^{i\beta_{n+1}}\theta_{n}\big)\hat{\Pi}\ket{\Psi_0} \nonumber\\
    =& \gamma_{n+1}\ket{\Psi_0}-\theta_{n+1}\hat{\Pi}\ket{\Psi_0},
\end{align}
where $\zeta = \braket{\Psi_0|\hat{\Pi}|\Psi_0}$ and
\begin{align}
\begin{cases} 
    \gamma_{n+1} =&  -\big(e^{-i\alpha_{n+1}}\gamma_{n}+A_{n+1}B_{n+1}\gamma_{n}\zeta\\
    &+e^{i\beta_{n+1}}A_{n+1}\theta_{n}\zeta\big)   ,  \\ 
    \theta_{n+1} =& - \big(\gamma_{n}B_{n+1}+e^{i\beta_{n+1}}\theta_{n}\big).
\end{cases}\label{eq:Gamma_Theta}
\end{align}
Then by induction, we can conclude that for any L, there exist $\gamma_{l},\,\, \theta_{l} \in \mathbb{C}$ satisfying Eq.~\eqref{eq:Gamma_Theta} such that the final state always lies in the subspace spanned by {$\ket{\Psi_0},\hat{\Pi}\ket{\Psi_0}$}.
The fidelity of the final state is 
\begin{align}
    \mathcal{F}_{\text{NOAA}} =& \braket{\Psi_f|\hat{\Pi}_{0}|\Psi_{f}}\\
                              =& (\gamma_{l}^* \bra{\Psi_0} - \theta_{l}^* \bra{\Psi_0}\hat{\Pi})\hat{\Pi}_{0}(\gamma_{l}\ket{\Psi_0}-\theta_{l}\hat{\Pi}\ket{\Psi_0})\\
                              =& |\gamma_{l}-\theta_{l}|^{2} \braket{\Psi_{0}|\hat{\Pi}_{0}|\Psi_{0}},
\end{align}
here we use the fact that $\hat{\Pi}_{0}\hat{\Pi}=\hat{\Pi}_{0}$.

\section{Proof of Lemma 2}\label{apendix:Lemma2}
Define
\begin{equation}
    \hat{Q} = (\hat{I} - \hat{\Pi}^{\textrm{mis}}_{0})\hat{\Pi}_{0}(\hat{I} - \hat{\Pi}^{\textrm{mis}}_{0}),
\end{equation}
\begin{equation}
    \hat{V} = \hat{\Pi}^{\textrm{mis}}_{0}\hat{\Pi}_{0}(\hat{I}-\hat{\Pi}^{\textrm{mis}}_{0}),
\end{equation}
then the ideal projector can be expanded as 
\begin{align}
    \hat{\Pi}_{0} =& \big(\hat{\Pi}^{\textrm{mis}}_{0}+ (1 -\hat{\Pi}^{\textrm{mis}}_{0})\big) \hat{\Pi}_{0} \big(\hat{\Pi}^{\textrm{mis}}_{0}+ (1 -\hat{\Pi}^{\textrm{mis}}_{0})\big)\nonumber \\ 
                  =& \hat{\Pi}^{\textrm{mis}}_{0}\hat{\Pi}_{0}\hat{\Pi}^{\textrm{mis}}_{0} +\hat{Q} + \hat{V} + \hat{V}^{^{\dagger}}
\end{align}
thus the fidelity of the mismatched NOAA is 
\begin{align}
    \mathcal{F}^{\textrm{mis}}_{\text{NOAA}}=& \braket{\tilde\Psi_{f}|\hat{\Pi}_{0}|\tilde\Psi_{f}}\nonumber \\
                                     =& \braket{\tilde\Psi_{f}|\hat{\Pi}^{\textrm{mis}}_{0}\hat{\Pi}_{0}\hat{\Pi}^{\textrm{mis}}_{0}|\tilde\Psi_{f}} + \braket{\tilde\Psi_{f}|\hat{Q}|\tilde\Psi_{f}} \nonumber \\ 
                                     &+ \braket{\tilde\Psi_{f}|\hat{V}|\tilde\Psi_{f}} + \braket{\tilde\Psi_{f}|\hat{V}^{\dagger}|\tilde\Psi_{f}} \nonumber\\
                                     =& \braket{\tilde\Psi_{f}|\hat{\Pi}^{\textrm{mis}}_{0}\hat{\Pi}_{0}\hat{\Pi}^{\textrm{mis}}_{0}|\tilde\Psi_{f}} +\Delta_{f}^{2}+2\mathrm{Re}(\braket{\tilde\Psi_{f}|\hat{V}|\tilde\Psi_{f}}),
\end{align}
where $\Delta_{f} = ||\hat{\Pi}_{0}(\hat{I} - \hat{\Pi}^{\textrm{mis}}_{0})\ket{\tilde{\Psi}_{f}}||$.
Using the Cauchy-Schwarz inequality, the norm of the term $\braket{\tilde\Psi_{f}|\hat{V}|\tilde\Psi_{f}}$ is bounded by
\begin{align}
    |\braket{\tilde\Psi_{f}|\hat{V}|\tilde\Psi_{f}}| =& |\braket{\tilde\Psi_{f}|\hat{\Pi}^{\textrm{mis}}_{0}\hat{\Pi}_{0}(\hat{I}-\hat{\Pi}^{\textrm{mis}}_{0})|\tilde\Psi_{f}}| \nonumber \\
    =&|\braket{\tilde\Psi_{f}|(\hat{\Pi}^{\textrm{mis}}_{0}\hat{\Pi}_{0}) \big(\hat{\Pi}_{0}(\hat{I}-\hat{\Pi}^{\textrm{mis}}_{0})\big)|\tilde\Psi_{f}}| \nonumber \\
    \le & || \hat{\Pi}_{0}(\hat{I}- \hat{\Pi}^{\textrm{mis}}_{0})\ket{\tilde\Psi_{f}}||\cdot ||\hat{\Pi}_{0}\hat{\Pi}^{\textrm{mis}}_{0}\ket{\tilde\Psi_{f}}||\\
    =& \Delta_{f}\sqrt{\braket{\tilde\Psi_{f}|\hat{\Pi}^{\textrm{mis}}_{0}\hat{\Pi}_{0}\hat{\Pi}^{\textrm{mis}}_{0}|\tilde\Psi_{f}}},
\end{align}
so the real part of $\braket{\tilde\Psi_{f}|\hat{V}|\tilde\Psi_{f}}$ is bounded by
\begin{align}
    -\Delta_{f}\sqrt{\braket{\tilde\Psi_{f}|\hat{\Pi}^{\textrm{mis}}_{0}\hat{\Pi}_{0}\hat{\Pi}^{\textrm{mis}}_{0}|\tilde\Psi_{f}}} \le  \mathrm{Re}(\braket{\tilde\Psi_{f}|\hat{V}|\tilde\Psi_{f}}) \nonumber\\ \le \Delta_{f}\sqrt{\braket{\tilde\Psi_{f}|\hat{\Pi}^{\textrm{mis}}_{0}\hat{\Pi}_{0}\hat{\Pi}^{\textrm{mis}}_{0}|\tilde\Psi_{f}}},
\end{align}
take this inequality into the expression of the fidelity, 
\begin{align}
    \braket{\tilde\Psi_{f}|\hat{\Pi}^{\textrm{mis}}_{0}\hat{\Pi}_{0}\hat{\Pi}^{\textrm{mis}}_{0}|\tilde\Psi_{f}} +\Delta_{f}^{2} &- 2 \Delta_{f}\sqrt{\braket{\tilde\Psi_{f}|\hat{\Pi}^{\textrm{mis}}_{0}\hat{\Pi}_{0}\hat{\Pi}^{\textrm{mis}}_{0}|\tilde\Psi_{f}}} \nonumber \\ 
    &\le \mathcal{F}^{\textrm{mis}}_{\text{NOAA}} \nonumber\le  \\
    \braket{\tilde\Psi_{f}|\hat{\Pi}^{\textrm{mis}}_{0}\hat{\Pi}_{0}\hat{\Pi}^{\textrm{mis}}_{0}|\tilde\Psi_{f}} +\Delta_{f}^{2} &+2\Delta_{f}\sqrt{\braket{\tilde\Psi_{f}|\hat{\Pi}^{\textrm{mis}}_{0}\hat{\Pi}_{0}\hat{\Pi}^{\textrm{mis}}_{0}|\tilde\Psi_{f}}}
\end{align}
which is equivalent to 
\begin{equation}
    (\sqrt{\mathcal{K}} - \Delta_{f})^2 \le \mathcal{F}^{\textrm{mis}}_{\text{NOAA}} \le (\sqrt{\mathcal{K}} + \Delta_{f})^2,\label{eq:Fidelity_of_mismatched_NOAA_bound}
\end{equation}
where $\mathcal{K} = \braket{\tilde\Psi_{f}|\hat{\Pi}^{\textrm{mis}}_{0}\hat{\Pi}_{0}\hat{\Pi}^{\textrm{mis}}_{0}|\tilde\Psi_{f}}$.

\section{Proof of Lemma 3}\label{apendix:Lemma3}
Using Yoder's phase, we first change our basis to an orthogonal one, let's define
\begin{align}
  \ket{t}=\frac{\hat{\Pi}\ket{\Psi_{0} } }{\sqrt{ \zeta }} ,\,\,\ket{\bar{t}}=\frac{(\hat{I}-\hat{\Pi})\ket{\Psi_{0} } }{\sqrt{ 1-\zeta }}   
\end{align}
where $\zeta=\braket{ \Psi_{0} |\hat{\Pi} |\Psi_{0}  }$ and $0<\xi<1$.
The initial state could be expressed as
\begin{align}
  \ket{\Psi_{0} }=\sqrt{ \zeta }\ket{t}+\sqrt{ 1-\zeta }\ket{\bar{t}}.    
\end{align}
It is easy to tell
\begin{equation}
  \braket{ t|\hat{\Pi} |t  }=1,\,\,\braket{ \bar{t}|\hat{\Pi} |\bar{t}  }  =0,  
\end{equation}
Within the two-dimensional subspace spanned by $\{\ket{t},\ket{\bar{t}}\}$, $\hat{\Pi}$ acts as the projector onto $\ket{t}$.
Thus the problem is reduced to the AA case, by using the fixed-point algorithm's phase\cite{yoderFixedPointQuantumSearch2014}, we have the final state
\begin{align}
\ket{\Psi_{f}}=&\left(\prod_{j}\hat{G}_{j}\right)\ket{\Psi_{0} }\\
=&\sqrt{ P_{L} }\ket{t}+e^{i\chi_{L}}\sqrt{ 1-P_{L} } \ket{\bar{t}}     
\end{align}
where $\chi_{L}$ is the relative phase between the two components, and 
\begin{align}
P_{L}=&\braket{ \Psi_{f}|\hat{\Pi} | \Psi_{f} } \\
=&1-\delta^{2}T_{L}(T_{1/L}(1 /\delta)\sqrt{ 1-\braket{ \Psi_{0}|\hat{\Pi} |\Psi_{0}  }  })^{2}.    
\end{align}
Using Eq. \eqref{eq:Fidelity_o f_NOAA_bound}, the NOAA fidelity is 
\begin{align}
\mathcal{F}_{\text{NOAA}}
= & \frac{\braket{ \Psi_{0}|\hat{\Pi}_{0} |\Psi_{0}  } }{\braket{ \Psi_{0}|\hat{\Pi} |\Psi_{0}  } }\nonumber\\
&\times\left[1-\delta^{2}T_{L}\left(T_{1/L}(1 /\delta)\sqrt{ 1-\braket{ \Psi_{0}|\hat{\Pi} |\Psi_{0}  }  }\right)^{2}\right]
\end{align}
\bibliographystyle{apsrev4-2}
\bibliography{ref}
\end{document}